\documentclass[aoas]{imsart}

\RequirePackage{amsthm,amsmath,amsfonts,amssymb,bbm}
\RequirePackage[authoryear]{natbib}
\RequirePackage[colorlinks,citecolor=blue,urlcolor=blue]{hyperref}
\RequirePackage{graphicx}

\usepackage[utf8]{inputenc}
\usepackage{subfigure}
\usepackage{enumerate}
\usepackage{scalerel,stackengine}

\usepackage{algorithm}
\usepackage{algpseudocode}

\startlocaldefs
\numberwithin{equation}{section}
\theoremstyle{plain}
\newtheorem{theorem}{Theorem}
\newtheorem{proposition}{Proposition}
\newtheorem{lemma}{Lemma}
\theoremstyle{definition}
\newtheorem{definition}{Definition}
\theoremstyle{remark}
\newtheorem{assumption}{Assumption}
\newtheorem{example}{Example}

\DeclareMathOperator*{\argmax}{\textup{argmax}}
\DeclareMathOperator*{\argmin}{\textup{argmin}}

\newcommand{\bX}{\mathbf{X}}
\newcommand{\bx}{\mathbf{x}}
\renewcommand{\P}{\mathbb{P}}
\newcommand{\pl}{\text{PL}}
\newcommand{\C}{\mathcal{C}}
\newcommand{\F}{\mathcal{F}}
\newcommand{\N}{\mathbb{N}}

\newcommand{\R}{\mathbb{R}}
\newcommand{\Chs}{\widehat{C}_{\textsc{hs}}}

\algdef{SE}[VARIABLES]{Variables}{EndVariables}
   {\algorithmicvariables}
   {\algorithmicend\ \algorithmicvariables}
\algnewcommand{\algorithmicvariables}{\textbf{global variables}}

\newcommand\reallywidehat[1]{%
\savestack{\tmpbox}{\stretchto{%
  \scaleto{%
    \scalerel*[\widthof{\ensuremath{#1}}]{\kern-.6pt\bigwedge\kern-.6pt}%
    {\rule[-\textheight/2]{1ex}{\textheight}}
  }{\textheight}%
}{0.5ex}}%
\stackon[1pt]{#1}{\tmpbox}%
}

\endlocaldefs

\begin{document}

\begin{frontmatter}
\title{Population based change-point detection for the identification of homozygosity islands}
\runtitle{Change-point detection of homozygosity islands}

\begin{aug}
\author[A]{\fnms{Lucas} \snm{Prates}},
\author[B]{\fnms{Renan} \snm{B. Lemes}}
\author[B]{\fnms{Tábita} \snm{Hünemeier}}
\and
\author[A]{\fnms{Florencia} \snm{Leonardi}}
\address[A]{Instituto de Matemática e Estatística, Universidade de São Paulo}

\address[B]{Instituto de Biociências, Universidade de São Paulo}
\end{aug}

\begin{abstract}
In this paper, we propose a new method for offline change-point detection on some parameters of the distribution of a random vector. We introduce a penalized maximum likelihood approach that can be efficiently computed by a dynamic programming algorithm or approximated by a fast greedy binary splitting algorithm. We prove both algorithms converge almost surely to the set of change-points under very general assumptions on the distribution and independent sampling of the random vector. In particular, we show the assumptions leading to the consistency of the algorithms are satisfied by categorical and Gaussian random variables. This new approach is motivated by the problem of identifying homozygosity islands on the genome of individuals in a population. Our method directly tackles the issue of identification of the homozygosity islands at the population level, without the need of analyzing single individuals and then combining the results, as is made nowadays in \emph{state-of-the-art} approaches. 
\end{abstract}


\end{frontmatter}


\section{Introduction}

In diploid organisms, such as humans, each individual's genome is organized into pairs of chromosomes, each half inherited from each parent. When an individual is an offspring of biologically related parents, both chromosomes of the same pair can share identical segments, creating long stretches of consecutive homozygosity, known as runs of homozygosity (ROH).

In the last decades, studies on the identification of ROH carried out in human populations have revealed the presence of ROH even in cosmopolitan non-inbred populations, disclosing an increment of inbreeding levels and the consequent reduction of genetic diversity of populations, which is proportional to the walking distance from Africa, as expected by the out-of-Africa model of human colonization \citep{ceballos2018runs,kirin2010genomic,lemes2018inbreeding,leutenegger2011consanguinity,pemberton2012genomic}.

The distribution of ROH along the chromosomes is very uneven, resulting in some genomic regions having significant absence (coldspots) or excess of ROH (ROH islands) \citep{ceballos2018runs}. The mechanisms for the emergence of these regions are still under discussion. For example, there is evidence that ROH islands could represent regions that harbor genes target of positive selection since low-recombination regions commonly are locations of selective sweeps, in which a new beneficial mutation increases in frequency and becomes fixed, causing the overall reduction in genetic diversity of the region \citep{ceballos2018runs,pemberton2012genomic}.

To detect ROH and ROH islands, the genetic material of individuals from a given population is genotyped, and a set of single nucleotide polymorphisms (SNPs) is obtained.
Each SNP entry is codified to $1$ if that SNP belongs to an ROH for that individual and to  $0$ otherwise, 
where a marker is defined as belonging to an ROH for an individual
if it is surrounded by a region with high frequency of homozygous SNPs.  
Finally, ROH islands are regions in which ROH are most frequent in that population. That is, the positions in the array with high frequency of individual ROH passing through them. Therefore, we can think on the problem of ROH island detection in a population as the identification of regions with high frequencies of 1's in the codified SNPs of individuals of that population. That is, this problem can be regarded as a change-point problem for the parameters of a multidimensional random vector with Bernoulli marginal distributions.

\subsection{Change point detection}
Classically, change point detection refers to the problem of determining the times at which sequential observed data 
undergoes an abrupt change.  In that type of setting, a change point may refer to changes in mean \citep{page1954,tsay1988,Keshavarz2018}, variance \citep{Chen1997,Hawkins2005}, regression slope \citep{Chow1960,Zhongjun2007}, general distributions forms \citep{matteson2014}, or  other types of change. 
Many of these methods have been applied to a wide range of problems such as stream anomaly detection in industry \citep{li2018application}, monitoring of sleep stages using EEG/EMG \citep{agudeloespana2020bayesian}, identification of cyberattacks on networks \citep{Tartakovsky2006}, between many other interesting applications. For more references about
change point methods for time series and other applications, we refer the reader to the book \citet{chen2011parametric} or the
recent reviews \citet{Truong2020} and \citet{Tze-San_Lee_2010}.

As \citet{niu2016} points out, from a methodological point of view, there is a variety of ways to formulate the change-point detection problem: online vs. offline, single change point vs. multiple change points, parametric vs. non-parametric, Bayesian vs. non-Bayesian, and many other approaches when we dive into specific estimation procedures.
But in most of the classical formulations for change point detection, the estimation problem is considered under the hypothesis that the number of observations along the dimension of interest grows. Early results of 
\cite{Hinkley(1970)} show that, even for more straightforward problems, these methods are usually consistent for the change points fractions, but not the change points properly. For offline issues, a possible interpretation is that we have a fixed number of series (usually one) over a unit interval. What grows with the number of observations, in this case, is the resolution over this interval.

As technology advances and other types of data arise from different application areas, new types of change point detection problems arise as well. The ROH islands detection for a population is an example of this. Indeed, here the dimension of interest is discrete and finite, and what grows is the number of independent observations of the random vector. Moreover, since the indices correspond to specific positions such as biological markers, it is essential to recover the exact location of the change points. Therefore, in this case, theoretical consistency results must guarantee convergence to the true change points as the number of independent observations of the joint multidimensional distribution grows. 

\subsection{Our contribution and related works}
In this work we consider several aligned, independent samples of a multidimensional distribution and we assume the distribution has a block structure with different parameters on each block. We focus on detecting such changes in the parameters, that is, the boundaries of the blocks when the number of observations grows. 

We propose a penalized maximum likelihood approach to detect the set of change points in the distribution. We allow for multiple change points without assuming an \emph{a priori} fixed, known number.
The penalized maximum likelihood approach has also been considered recently in \citet{castro2018model, leonardi2021independent}, but on a different type of change-point problem. There, the approach was introduced 
for non-parametric discrete distributions in order to detect points of independence on a multidimensional random vector, under independent or non-independent sampling.
Here we use a similar methodology but on a different problem; namely, the detection of changes on some parameters of the distribution of the random vector.
We also show how to compute the exact estimator by a dynamic programming algorithm and how to approximate the global solution by an efficient greedy binary segmentation algorithm. Albeit using a parametric approach, we do not consider a specific family of distributions. Instead, we state a set of general hypotheses that the family must satisfy and then prove the consistency of the change point set estimator when the number of aligned sequences grows, for the exact estimator as well as for the binary segmentation algorithm. The approach proposed here naturally allows us to directly use bootstrap techniques to assess the variability of the change point estimators. Bootstrap resampling methods have been applied before to construct confidence intervals for the change point location in the single change point problem   \citep{cpd_and_boot,Hu_kov__2008}  or on multiple change point problems \citep{frick2013multiscale}.
We suggest using the bootstrap method to approximate the probability of each index being detected as a change point and to measure the variability of the entire set of change points. A simulation study is presented to provide insights on the behavior of the algorithms and the bootstrap confidence estimation as sample size grows. Finally, we discuss the identification of the ROH islands on the African and European populations, based on the analysis of SNP data from the Human Genome Diversity Project (HGDP) \citep{li2008worldwide}. 
We compare the ROH islands detected by our method with those detected by the procedure proposed in  \citet{mcquillan2008runs,kirin2010genomic}, using the PLINK software \citep{purcell2007plink}, a well-established command line software designed to solve many medical and population genetics problems. 

It is worth mentioning that the methodology for change-point detection developed in this work is very general, and it is not restricted to discrete distributions. In particular, we verify the conditions leading to the consistency of the estimators for standard distributions such as categorical random variables and Gaussian random variables. Still, other families of distributions can satisfy these conditions, widening the possibility of applying our method to many different applied problems.



\section{Change Point Estimation}

\subsection{Definition of the model}

Let $m$ be a positive integer and consider a random vector $\bX=(X_1,\dots, X_m)$ taking 
values in $A^m$, where $A$ is any state space. Let $\Theta$ be a parametric space and consider a family $\mathcal{F} =  
\{f_\theta\,|\, \theta\,\in\,\Theta\}$ of probability densities or probability mass functions over $A^\ell$, for each $\ell\in\N$, indexed by $\theta\in\Theta$. We assume each variable $X_i$, $i=1,\dots,m$  
is distributed according to some $f_{\theta}\in \mathcal F$. Given two integers $r$ and $s$, with 
$r \leq s$, we use the notation $r:s$ for the set 
$\{r,r + 1\ldots,s\}$. Denote by $\C$ the class of all ordered sets
$C = \{c_0, c_1,\ldots, c_{k}\} \subseteq 0:m$ such that $c_0 = 0$ and $c_{k} = m$. 
We say that $C\in\C$ is a change point set for the random 
vector $\mathbf{X}$ if the variables in the subvector $(X_{c_{j-1}+1},\ldots,X_{c_j})$ are
identically distributed with distribution $f_{\theta_{j}}\in \mathcal F$ 
and $\theta_j\neq \theta_{j+1}$ for all $j=1,\dots, k-1$. Observe that the change point vector is
unique, and from now on it will be denoted by $C^*$. Given any set $C\in\C$, we denote by $k_C$ the number of positive elements in $C$; that is, if $C = \{c_0, c_1,\ldots, c_{k}\}$ then $k_C = k$. 

Assuming the blocks  in the random vector $\bX$ are independent, we can write the probability of 
observing $\bx=(x_1,\dots, x_m) \in A^m$ in the model with parameters $(C,\theta)$, $\theta=\{\theta_j\in \Theta\colon j\in 1:k_C\}$
by
\begin{equation}\label{prob}
\P_{(C,\theta)}(\bx) = \prod_{j=1}^{k_C}
f_{\theta_j}(x_{(c_{j-1}+1):c_j})\,,
\end{equation}
where $f_{\theta_j}(x_{(c_{j-1}+1):c_j})$ represents the distribution of a random vector over $A^{c_j-c_{j-1}}$ with parameter $\theta_j$. The independence assumption over the different blocks is
not a necessary condition for the method but the generalization to a non-independent setting is out of the scope of this work. We present below two well-known examples of families of distributions, one for continuous random variables and the other for discrete random variables, with the exact expression for the likelihood function \eqref{prob}.

\begin{example}
\label{example_categorical}
Let $A=\{a_1,\dotsc, a_d\}$ be a finite set and let
\[
\F=\Bigl\{ p\in [0,1]^d\colon p(a_i)\geq 0\text{ and } \sum_{i=1}^{d} p(a_i)=1\Bigr\}
\]
be the family of all  probability distributions over $A$. The likelihood function \eqref{prob} 
for $\bx\in A^m$ can be written, for $C\in\C$ and $\theta=(p_1,\dots,p_k)$ as 
\begin{equation*}
\P_{(C,\theta)}(\bx) = \prod_{j=1}^{k_C}\;\prod_{c = c_{j-1}+1}^{c_j} p_j(x_c)\,.
\end{equation*}
\end{example}

\begin{example}
\label{example_gaussian}
Consider the family 
\begin{equation}
    \mathcal{F} = \{f_{(\mu,\sigma^2)} \colon (\mu, \sigma^2)\,\in\,\mathbb{R}\times\mathbb{R}_{+} \} \label{family_normal}
\end{equation}
of probability densities over $A=\mathbb R$,  where $f_{(\mu, \sigma^2)}$ denotes the density of a Gaussian random variable  with mean $\mu$ and variance $\sigma^2$. Fixing the change point set $C$ with $k_C = k$ and the vector of parameters $\theta = ((\mu_j, \sigma_j^2))_{j=1}^{k}$, the likelihood function assuming the variables in $\bx$ are independent is given by
\begin{equation*}
    \mathbb{P}_{(C,\theta)}(\mathbf{x}) = \prod_{j=1}^{k_C}\;\prod_{c=c_{j-1}+1}^{c_j} (2\pi)^{-1/2}\sigma_j^{-1}\exp\bigl\{- \frac{1}{2} \Bigl( \frac{x_c - \mu_j}{\sigma_j}\Bigr)^2\bigr\} \,.
\end{equation*}
\end{example}

\subsection{Estimation and model selection}

Consider a sample of $n$ i.i.d. random vectors $\bx^n= \{\bx^{(i)}\}_{i=1}^n$ distributed as $\bX$, with 
change point set  $C^*=(c^*_0,\dots, c^*_{k^*})$ and parameters $\theta^*=(\theta^*_1,\dots, \theta^*_{k^*})$. Our main goal is to  
estimate the change point set $C^*$ and the parameters $\theta^*$. 


For any integer interval 
$I\subset 1:m$ assume we can compute the maximum likelihood estimator based on the  subsample $\{x^{(i)}_j\}_{i\in 1:n, j\in I}$. Write the maximum likelihood function for the set $C$ of candidate change points as \begin{equation}\label{likelihood}
L(C ;\bx^n) \;=\; \prod_{i = 1}^{n}\,
\P_{(C,\widehat\theta)}(\bx^{(i)})\;=\; \prod_{i = 1}^{n}\,
\prod_{j=1}^{k_C} f_{\widehat\theta_j}\left(x_{(c_{j-1}+1):c_j}^{(i)}\right) 
\end{equation}
where $\widehat\theta_j$ denotes the maximum likelihood estimator computed on the sample 
$\{x^{(i)}_c\}_{i\in 1:n, c\in I_j}$ with $I_j= (c_{j-1}+1):c_{j}$.
From \eqref{likelihood} the log-likelihood function  is given by 
\begin{equation}\label{log-likelihood}
l(C;\bx^n) = \sum_{i = 1}^{n}\sum_{j=1}^{|C|} 
\log f_{\widehat\theta_j}\bigl(x_{(c_{j-1}+1):c_j}^{(i)}\bigr) \,.
\end{equation}

Let $R\colon \mathcal{C} \to \mathbb{R}$ denote some regularization function and $J\colon \N \to \mathbb{R}$ an increasing function on the sample size $n$. We introduce the penalized likelihood estimator based on the functions $R$ and $J$ in the following definition. 
	
\begin{definition}
Given a sample $\bx^n$ and a constant $\lambda>0$, the Penalized Likelihood (PL) function for the set of change points $C$ is defined as
\begin{equation}\label{pl}
\pl(C;\bx^n) \;=\; -l(C; \bx^n) + \lambda R(C)J(n) \,.
\end{equation}
The PL estimator for the change point set is then defined as 
\begin{equation}\label{hatC}
\widehat{C}(\bx^n)\; =\; \argmin_{C \in \C} \pl(C;\bx^n)\,.
\end{equation}
\end{definition}

As we will show later in Theorems~\ref{PL_Consistency} and \ref{HS_Consistency}, in order to obtain the consistency of the change point estimator defined by \eqref{hatC} we need the functions $R$ and $J$ to satisfy some properties. This will be made precise in the statements of these theorems.

\section{Computation of the Estimator}\label{sec:computation}

In order to efficiently estimate the change point set, we suppose in this section that the regularization function $R$ is additive. That  means that there exists a function $\rho\colon \{1,\dots,m\}^2 \to \mathbb{R}$ such that for $C = \{c_0,\ldots,c_k\}$ we have 
\begin{equation}\label{Radd}
R(C) \;=\;  \sum_{j=1}^{k_C} \rho(c_{j-1}+1, c_j)\,.
\end{equation}

\subsection{Dynamic programming segmentation algorithm}

As presented in \citet{Jackson2005}, dynamic programming can be used to calculate exactly the 
PL estimator. Under an additive regularization, the function we want to minimize can be written as
\begin{align*}
 -l(C;\bx^n) + \lambda J(n)R(C) 
 &= \sum_{j=1}^{k_C} -\log f_{\widehat\theta_j}\Bigl(x_{(c_{j-1}+1):c_j}^{(i)}\Bigr) + \lambda J(n)\rho(c_{j-1} + 1, c_{j})\\ 
  &=  \sum_{j=1}^{k_C} Q((c_{j-1}+1):c_{j}) \,, 
 \end{align*}
where 
\begin{equation*}
Q((c_{j-1}+1):c_{j}) = -\log f_{\widehat\theta_j}\Bigl(x_{(c_{j-1}+1):c_j}^{(i)}\Bigr) + \lambda J(n)\rho(c_{j-1} + 1, c_{j}) \,.
\end{equation*}

The equation shows that we can completely decouple the loss from different blocks. Let $\mathcal{C}_i$ be the set of all ordered change point sets in $1:i$. Define
\begin{equation*}
F(i) = \min_{C \, \in \, \mathcal{C}_i} \Bigg\{\sum_{j=1}^{k}  Q((c_{j-1}+1):c_{j})\Bigg\}  
\end{equation*} 
as the optimal value for the segmentation up to variable $i$. The estimator $\widehat C(\bx^n)$ given by \eqref{hatC} is obtained by computing $F(m)$. But notice that
\begin{equation*}
F(i) = \underset{c \in (k-1):(i-1)}{\min}\{F(c) + Q((c+1):i)\} \,,
\end{equation*}
which establishes a recursion equation for the values of $F(i)$, $i$ varying from $1$ to $m$. The value of $F(1)$ can be computed trivially, and then we use the recursion to compute the values until we reach $F(m)$.

Albeit $m$ is fixed and only $n$ grows, the number of variables $m$ can be very large in some applications, so it is useful to express the complexity in terms of both. 
The dynamic programming segmentation algorithm runs on a time complexity of $O(m^2)$. However, we are assuming that $Q$ have been previously computed for all intervals in $1:m$. 

To compute $Q$, we need to compute the maximum likelihood estimators for each block. For most models, this can be done efficiently by computing the sufficient statistics, and then compute the entries of $Q$. For the models presented in Examples~\ref{example_categorical} and \ref{example_gaussian}, the time complexity for computing the sufficient statistics is $T(n, m) = O(nm)$, and then $O(m^2)$ to compute the entries of $Q$. In the case where no sufficient statistics exist, we need to reprocess the data every time, so the complexity to compute $Q$ is $O(nm^3)$.

Hence, the final time complexity of the algorithm is $O(T(n, m) + m^2)$. In the worst case the algorithm is $O(nm^3)$, and can be very slow for big values of $m$.

Depending on the function $\rho$ chosen, the PELT algorithm of \cite{Killick_2012} might be applicable. It consists of an adaptation of the dynamic programming algorithm discussed here. In some scenarios, such as when the number of change points is proportional to the number of variables, it runs in $O(m)$. The final complexity would be $O(nm + m)$ when suitable sufficient statistics exist.

\subsection{Hierarchical segmentation algorithm}

For efficient computation of $\widehat C(\bx^n)$ we can use an approximation to the optimum in \eqref{hatC}, known as  hierarchical segmentation or binary segmentation, first proposed in \cite{scott1974cluster}. Given an integer interval $I = r:s$, write $\mathbf{x}_I^n$ for the data with columns restricted to $I$ and assume the penalizing function $R$ is additive. Define the penalized loss of $I$ as

\begin{equation*}
\pl(I) \;=\; \pl(\{0,s-r+1\}; \mathbf{x}_I^n) \;=\; -l(\{0,s-r+1\}; \mathbf{x}_I^n) +\lambda\rho(r,s)J(n) \,,
\end{equation*}
with the appropriate renumbering of the columns in $\bx_I^n$. That is, the penalized loss corresponding to the interval $I$ is the penalized loss defined in \eqref{pl} when we only consider the data $\mathbf{x}_I^n$ and perform no splits. We use the convention $\pl(\emptyset) = 0$. For $c \, \in \, I=r:s$, define
\begin{equation}\label{hI}
h_I(c) \;=\; \pl(r:c) + \pl((c+1):s) \,.
\end{equation}
Observe that when $c=s$, by convention we have $(s+1):s = \emptyset$ so that $h_I(s)= \pl(I)$.

The hierarchical segmentation algorithm works recursively as follows. It begins with the set $\Chs(\bx^n)=\{0,m\}$ corresponding to the single interval $I=1:m$. In each iteration and for  each interval $I$ determined by $\Chs(\bx^n)$, the algorithm computes 
\[
\hat c \;=\; \underset{c \in I}{\argmin} \; h_I(c)
\]
and adds it to $\Chs(\bx^n)$. Observe that if in one iteration $\hat c=s$, as $s\in \Chs(\bx^n)$, no changes are made on $\Chs(\bx^n)$. The process continues until no more points can be added to 
$\Chs(\bx^n)$. 


To evaluate $\pl$ at each interval, we either store all possible values in the same fashion as for the dynamic programming algorithm or we evaluate them on the run. Since storing would make the algorithm $O(m^2)$ in any scenario, it is more interesting to pre-compute the sufficient statistics and evaluate the loss on the intervals as they appear.


In the worst-case scenario, the algorithm has time complexity of order $O(T(n, m) + m^2)$. However, the algorithm has asymptotic complexity of order  $O(T(n, m) + mk_{C^*}))$, as stated in the next proposition.

\begin{proposition}
    \label{prop: hier_complexity}
        The hierarchical segmentation algorithm  will asymptotically perform exactly $2k_{C^*} - 1$ recursive calls. Moreover, its asymptotic complexity is of order  $O(T(n, m) + mk_{C^*})$, where $T(n, m)$ is the time complexity to compute the sufficient statistics for each interval.
\end{proposition}

The proof of this proposition is given in Appendix~\ref{proof: hier_complexity}.    
    
\section{Consistency of the Change Point Estimators}\label{sec:consistency}

In this section we state the theoretical results  
that guarantee the consistency of the estimator \eqref{hatC} computed exactly by the dynamic programming algorithm or approximated by the hierarchical segmentation algorithm. For each method, we state a different set of assumptions that must be satisfied by the family of probability distributions considered in the model. 
Then we prove that standard families of distributions such as Gaussian and categorical probability distributions satisfy these assumptions. 
    
\begin{assumption}\label{ass_PL}
Suppose there exists a function $l^*\colon \C\to\R$ such that

\begin{enumerate}
\item[(PL1)] For any $C\in\C$ we have that 
\[
\frac1n l(C;\bx^n) \;\to \; l^*(C)
\]
almost surely as $n\to\infty$, where $l$ is the log-likelihood function defined in \eqref{log-likelihood}. Moreover, assume there exists $\alpha>0$ such that 
\[
\inf_{C\not\supseteq C^*} l^*(C) \;\leq\; l^*(C^*) - \alpha\; <\; l^*(C^*)\,.
\]
\item[(PL2)] There exists a sequence $\{v(n)\}_{n\in\N}$ 
satisfying $v(n) \to \infty$ and $v(n)/n \to 0$ when $n\to\infty$, and such that for any $C\supseteq C^*$ we have that 
\begin{equation*}
 \big|l(C; \bx^n) - l(C^*;\bx^n)\big| \;<\; v(n)
\end{equation*}
eventually almost surely as $n\to\infty$. 
\end{enumerate}
\end{assumption}

Observe that (PL2) implies that $l^*(C)=l^*(C^*)$ for all $C\supseteq C^*$. \\
    
We now state the consistency result of the PL estimator
given in \eqref{hatC}.
    
    
\begin{theorem}
\label{PL_Consistency}
Suppose that the family $\mathcal F$ of probability distributions satisfy Assumption~\ref{ass_PL}.
Let $R$ be a penalizing function such that 
 $R(C) > R(C')$ whenever $C\supset C'$ and let $J(n)$ be such that  $J(n)/v(n) \to\infty$ and $J(n)/n \to 0$ when $n\rightarrow \infty$. Then the estimator of the change point set given by \eqref{hatC} satisfies $\widehat{C}(\bx^n)=C^*$ eventually almost surely as $n\to\infty$.
\end{theorem}
    
Notice that the regularization function $R$ does not need to be additive to guarantee the consistency of the PL estimator. However, this is a desirable property to efficiently compute the estimator by using the dynamic programming segmentation algorithm. 

In order to prove that the estimator given by the  hierarchical segmentation algorithm is also consistent, we need a slightly different set of hypotheses considering the local nature of this algorithm. 
Denote by $\mathcal I$ the set of all intervals $I\subset 1:m$. Given $I\in\mathcal I$, denote by $\widehat{\theta}_I$  the maximum likelihood estimator of the parameter $\theta$ on the interval $I$ and as before, let $\bx_I^n$ be the data restricted to the interval  $I$. Define the maximum log-likelihood function for the interval $I$ as
\begin{equation*}
l(I;\bx_I^n) = \sum_{i=1}^n\log f_{\widehat{\theta}_I}(x_{I}^{(i)}) \,.
\end{equation*}
Let $h_I$ be the loss function for the interval $I$ as defined in \eqref{hI}.

\begin{assumption} \label{ass_HS}  
There exists a function $l^*\colon\mathcal I\to\mathbb R$ such that
\begin{enumerate}
\item[(H1)] For any integer interval $I\in\mathcal I$ we have that 
\begin{equation}\label{lstar}
\frac1n l(I;\bx_I^n) \;\to \; l^*(I)
\end{equation}
almost surely as $n\to\infty$. If $I=r:s$,  defining $h_I^*:I\to \mathbb{R}$ as
\begin{equation}\label{hstar}
h^*_I(u) \;=\; -l^*(r:u) - l^*((u+1):s)\,,\quad  u\,\in\,I\,,
\end{equation}
we have that  
\begin{equation}
 \min_{c \in I\setminus\{s\}\cap C^*} \; h_I^*(c)\;<\; \min_{c \not\in I\setminus\{s\}\cap C^*} \; h_I^*(c) 
\label{hyp_strict_max}
\end{equation}
\item[(H2)] There exists a sequence $\{v(n)\}_{n\in\N}$ satisfying $v(n) \to \infty$ and $v(n)/n \to 0$ when $n\to\infty$, and such that, for any integer interval $I=r:s$ satisfying $I\setminus\{s\} \cap C^* = \emptyset$ we have
\begin{equation*}
\max_{u\in I} \;\big|l(I; \bx_I^n) - l(r:u;\bx_I^n) - l((u+1):s;\bx_I^n)\big| \;<\; v(n)
\end{equation*}
eventually almost surely as $n\to\infty$.
\end{enumerate}
\end{assumption}

We can now state the consistency of the change point set estimator given by the hierarchical segmentation algorithm.
    
\begin{theorem}
\label{HS_Consistency}
For any $\lambda>0$, let $\Chs(\bx^n)$ be the estimator computed by the hierarchical segmentation algorithm. Suppose that the family $\mathcal F$ of probability distributions satisfy Assumption~\ref{ass_HS}.
Suppose that $R$ satisfies \eqref{Radd} for some function
 $\rho\colon \mathcal I \to \mathbb{R}$ and that 
$\rho(I) < \rho(I_1)+\rho(I_2)$ whenever $I=I_1\cup I_2$, with $I_1,I_2\neq \emptyset$. 
Finally, assume that the function $J(n)$ satisfies $J(n)/v(n) \to\infty$ and $J(n)/n \to 0$ when $n\rightarrow \infty$. Then, $\Chs(\bx^n) = C^*$ eventually almost surely as $n\to\infty$.
\end{theorem}

It turns out that verifying directly the inequality \eqref{hyp_strict_max} on Assumption~\ref{ass_HS} is usually hard. We now state a lemma that provides easier to verify conditions that imply the desired inequality.

\begin{lemma}\label{conditions_ineq_HS}
Given an integer interval $I$ 
assume that 
\begin{enumerate}
\item[(a)] There exists a point $c\,\in\,I$ such that $ h_I^*(c) < h_I^*(s)$. 
\item[(b)] The function $h_I^*$ restricted to the intervals $[c^*_{j-1},c^*_j]$, $j=1,\dotsc, k^*$,  is concave. 
\item[(c)] If $h_I^*$ is constant in $[c_{j-1}^*, c_{j}^*]$, then it is equals to $h_I^*(s)$ on this interval. 
\end{enumerate}    
Then, the function $h^*_I$ satisfies that  
\begin{equation*}
 \min_{c \in I\setminus\{s\}\cap C^*} \; h_I^*(c)\;<\; \min_{c \not\in I\setminus\{s\}\cap C^*} \; h_I^*(c) \,. 
\end{equation*}
\end{lemma}

The following lemma shows that under general conditions on the function $h^*_I$ verified in practice by many models, the conditions of Lemma~\ref{conditions_ineq_HS} hold.

\begin{lemma}\label{conditions_iid_HS}
Let $\Theta$ be a convex subset of $\mathbb R^d$ for some $d\in\N$. For any interval $I\in\mathcal I$ define the parameter $\theta_{I} = \sum_{k=1}^{k^*} \frac{|I \cap I_k^*|}{|I|} \theta_r^*$. Assume the function $h^*_I$ defined by \eqref{hstar}, for $I=r:s$,  is of the form $h^*_I(u) = (u-r+1)\psi(\theta_{r:u}) + (s-u)\psi(\theta_{(u+1):s}) + \alpha$,  with $\psi$ a strictly convex function having second order derivatives in the interior of $\Theta$ and $\alpha$ a constant not depending on $u$. Then $h^*_I$ satisfies the conditions (a)-(c) in Lemma~\ref{conditions_ineq_HS}.
\end{lemma}

We now state two results that guarantee that both families of probability distributions given in Examples~ \ref{example_categorical} and \ref{example_gaussian} satisfies Assumptions~\ref{ass_PL} and \ref{ass_HS}. Hence, both the dynamic programming and hierarchical segmentation algorithms provide  consistent estimators of the change point parameters $(C^*,\theta^*)$.   

\begin{proposition}\label{prop_categorical}
The family of discrete categorical distributions, as presented in Example~\ref{example_categorical}, satisfies Assumptions~\ref{ass_PL} and \ref{ass_HS}.
\end{proposition}

\begin{proposition}\label{prop_gaussian}
The family of Gaussian densities, as presented in Example~\ref{example_gaussian}, with unknown mean and constant known variance for all 2blocks satisfies Assumption~\ref{ass_PL} and \ref{ass_HS}.
\end{proposition}

The proof of Proposition~\ref{prop_categorical} is presented in Appendix~\ref{appendix:proof_of_consistency} and the proof of Proposition~\ref{prop_gaussian} is given in the Supplementary Material to this article, jointly with results of simulations implemented for this specific family of distributions. 

\section{Bootstrap Methods to Assess Variability}\label{sec:bootstrap}

A criticism usually made to change point detection methods is that they often do not provide measures of variability or confidence sets for the estimated change points set.  Here, we propose the usage of bootstrap methods to tackle this problem.


As mentioned in the Introduction, in our setting, each index usually has a specific interpretation, such as being a genetic marker or a particular position on a river. Hence, a researcher might be interested in the statistical significance of a given point detected as a change point. To know how likely an index $c$ in $1:(m - 1)$ is detected as a change point on a model with parameters $(C^*, \theta^*)$  we  can compute
\begin{equation*}
p(c) = \mathbb{P}_{(C^*,\theta^*)}\bigl(\, c \, \in \, \widehat{C} \,\bigr) \,.\quad 
\end{equation*}
Since our scenario assumes the rows of $\bx^n$ correspond to i.i.d samples of the vector $\bX$, bootstrap resampling is straightforward. Resampling $B$ data sets $\bx^{n,1},\dots, \bx^{n,B}$ with the same dimension of $\bx^n$ by using the bootstrap on the rows of $\bx^n$ and computing the change point set $\widehat{C}(\bx^{n,b})$ for $b=1,\dots,B$, the bootstrap estimate for the quantity above is
\begin{equation*}
\hat{p}(c) \;=\; \frac{1}{B}\sum_{b=1}^B  \mathbbm{1}_{\widehat{C}(\bx^{n,b})}(c)\,,
\end{equation*}
where $\mathbbm{1}_{\widehat{C}(\bx^{n,b})}(c)$  denotes the indicator function that $c\in \widehat{C}(\bx^{n,b})$. 
In practice, however, it might be difficult for the researcher to pinpoint the exact location of the change point beforehand. His practical experience most likely indicates that a change occurs within an region. Then,
we can address the question of how significant it is for the interval $I$ to  contain a change point  by computing

\begin{equation*}
p(I)\;=\;  \mathbb{P}_{(C^*,\theta^*)}\bigl( \,I \cap \widehat{C} \neq \emptyset\,  \bigr) \,.
\end{equation*}
This quantity can be estimated by
\begin{equation*}
\hat{p}(I) = \frac{1}{B}\sum_{b=1}^B  \mathbbm{1}_{\{\, I\cap \widehat{C}(\bx^{n,b}) \,\neq\, \emptyset \,\}}\,.
\end{equation*}


To measure the variability of the whole estimated change point set, we could consider a metric $d$ defined on sets in $\mathcal C$. Our interest would be to calculate
\begin{align*}
\mathbb{E}_{(C^*,\theta^*)} \bigl[d(C^*,\widehat{C})\bigr] \quad \mbox{ and } \quad
\mbox{Var}_{(C^*,\theta^*)}\bigl[d(C^*,\widehat{C})\bigr] \,.
\end{align*}
By replacing $C^*$ by $\widehat C(\bx^n)$, the bootstrap estimates of these quantities are 
\[
\bar d \;=\;\frac{1}{B}\sum_{b=1}^Bd(\widehat{C}(\bx^n),\widehat{C}(\bx^{n,b}))
\quad \mbox{and} \quad \widehat{\text{Var}}(d)\;=\; \frac{1}{B}\sum_{b=1}^B\bigl[\, d(\widehat{C}(\bx^n),\widehat{C}(\bx^{n,b})) - \bar d\,\bigr]^2 \,,
\]
respectively. 
    
\section{Simulations}

In this section we show the results of a simulation study  to compare the convergence of both  algorithms as sample size increases. We consider two families of distributions: the Bernoulli distribution and the Gaussian distribution with both mean and variance unknown (results available as Supplementary Material).  
We fixed the number of variables $m$ and varied the number of  change points in the model.


    
\begin{figure}
\centering
\includegraphics[width = 0.49\textwidth]{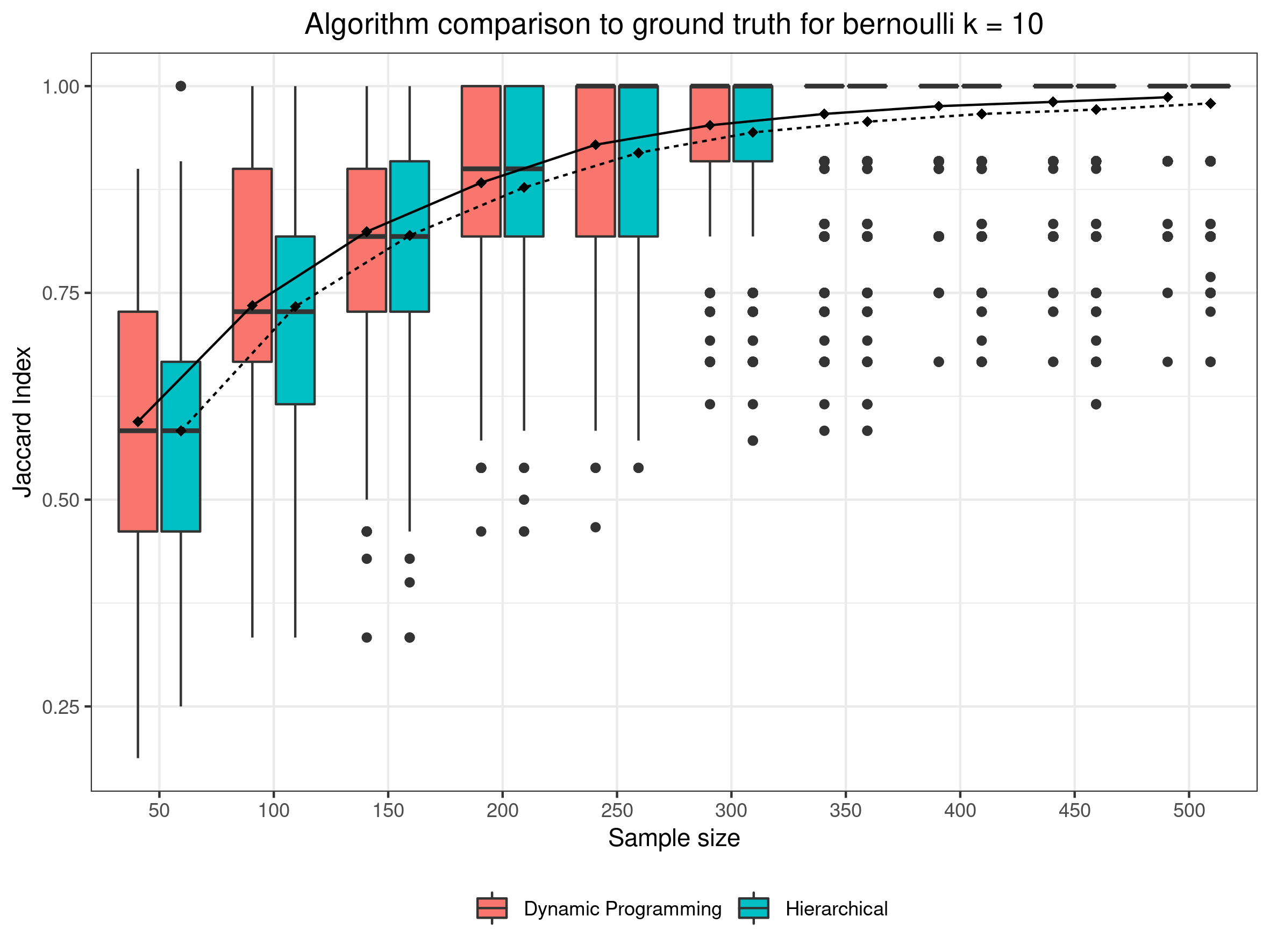}
\includegraphics[width = 0.49\textwidth]{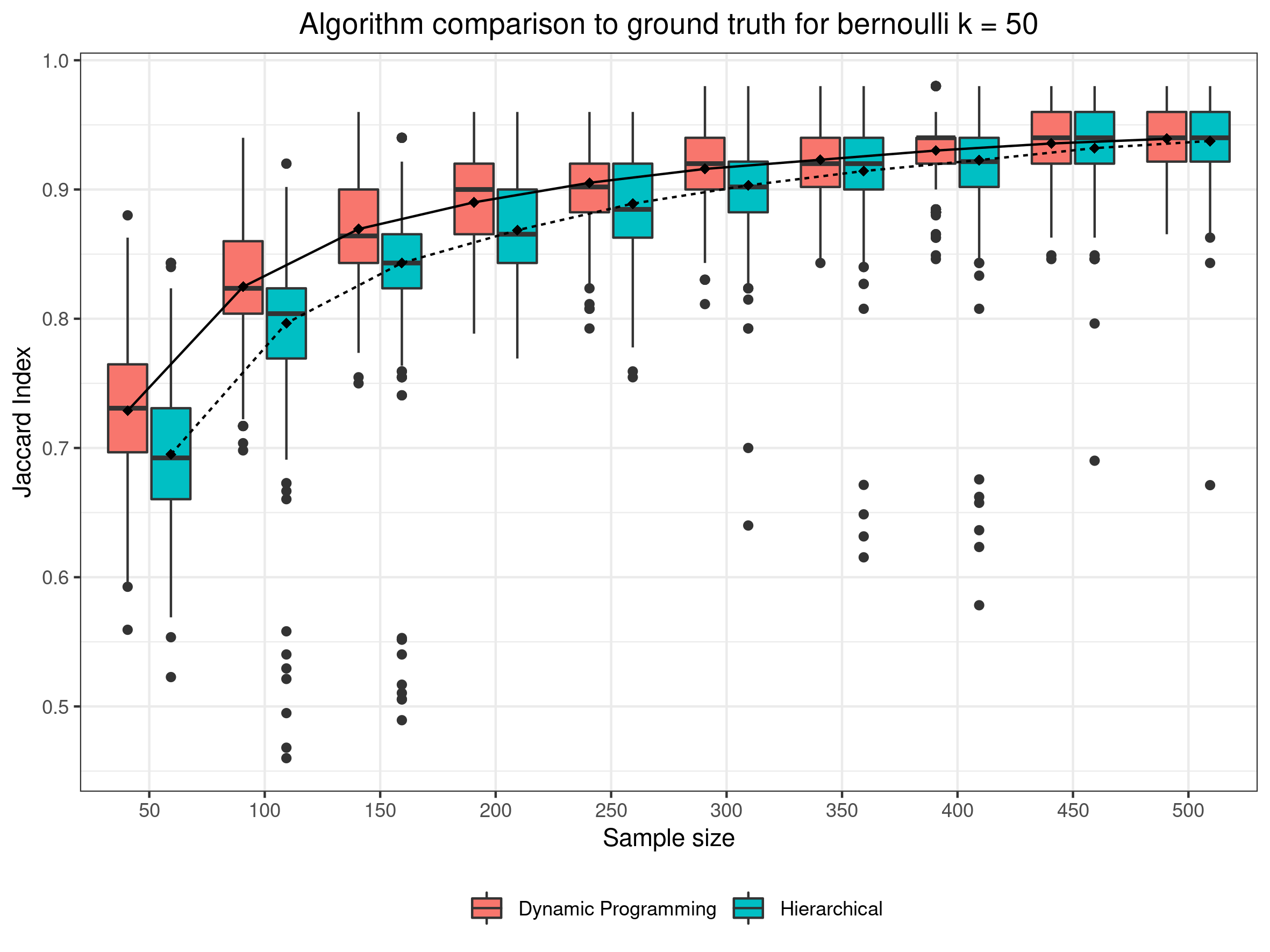}\\[3mm]
 \includegraphics[width = 0.49\textwidth]{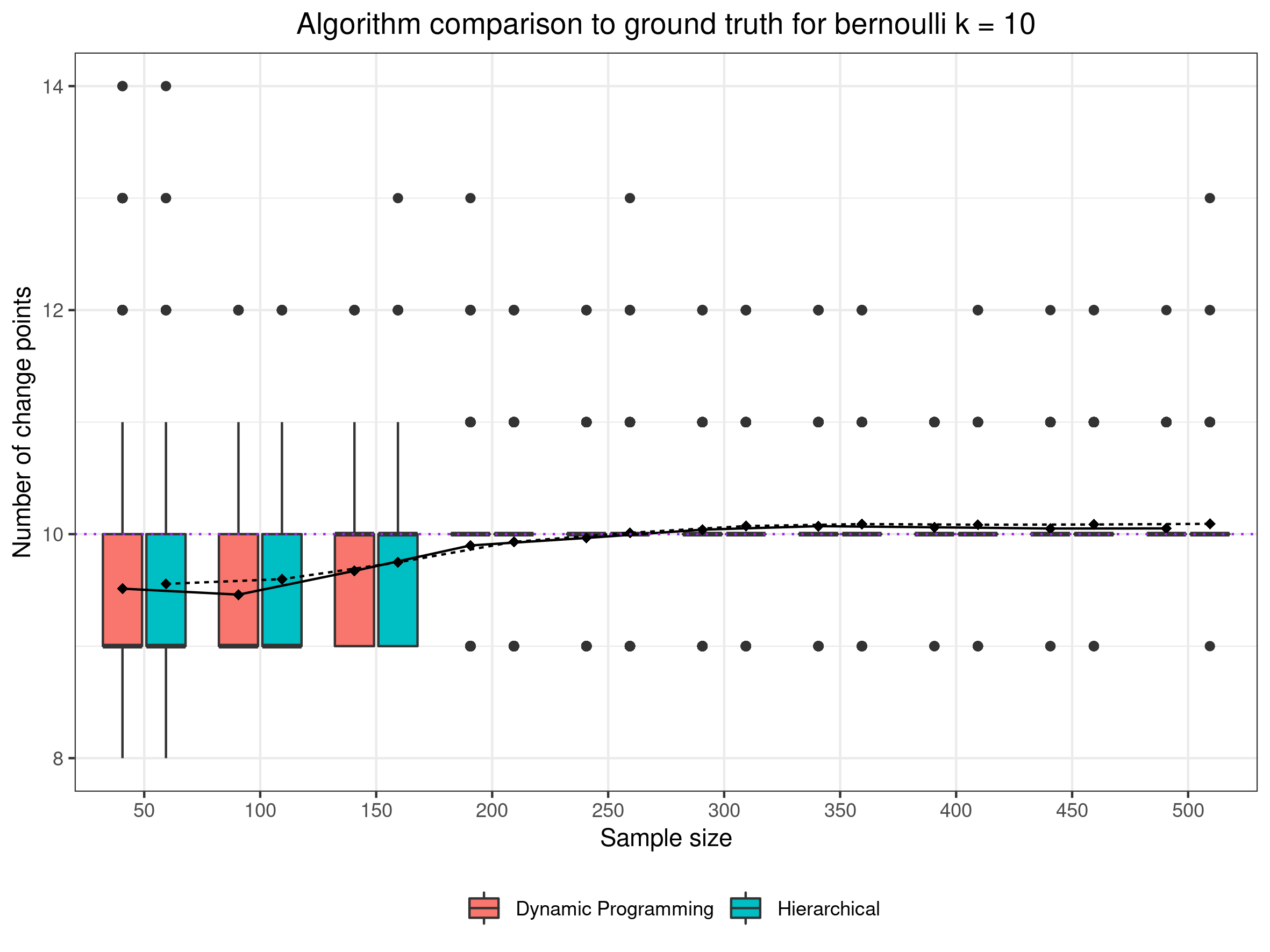}
  \includegraphics[width = 0.49\textwidth]{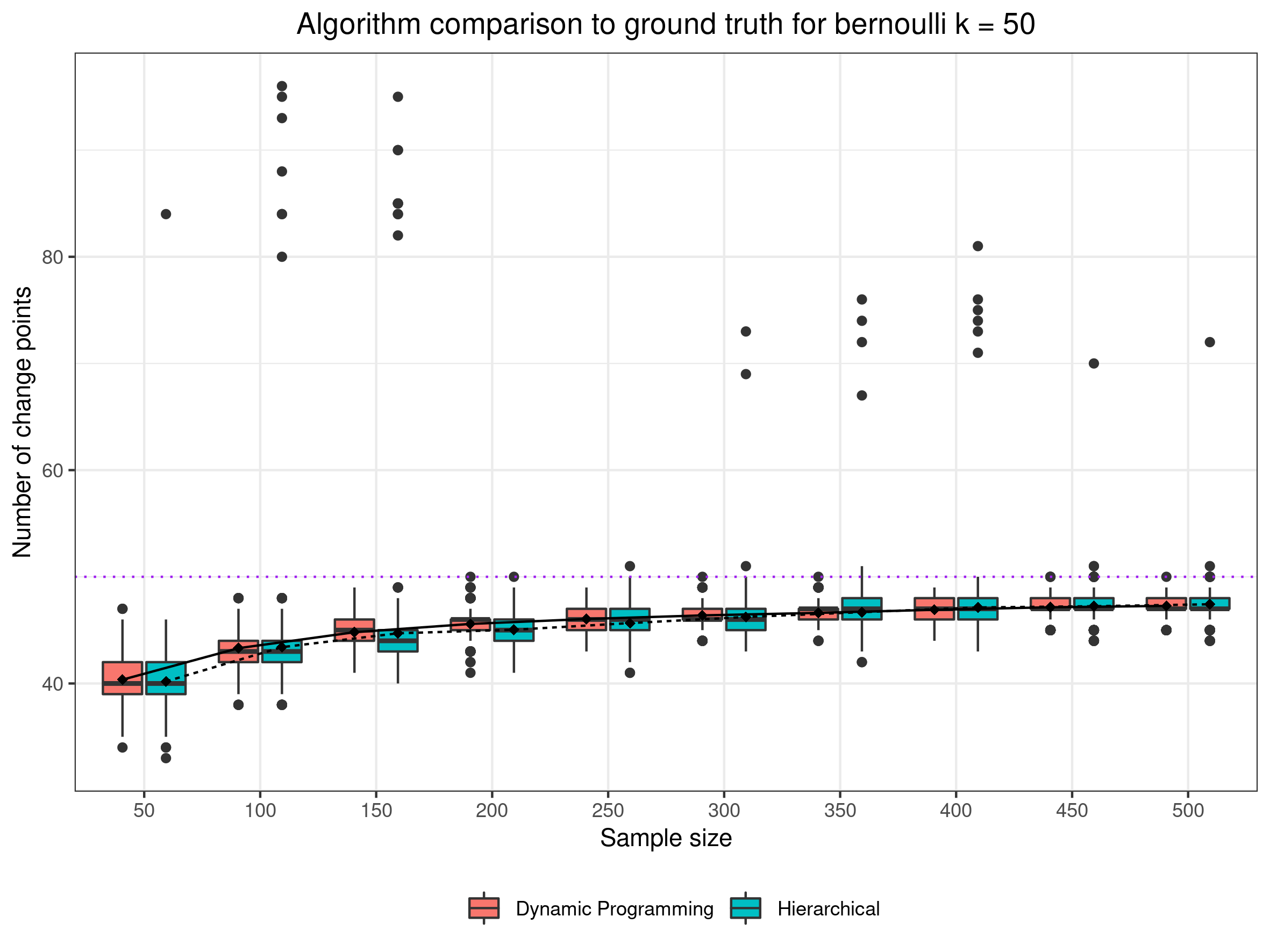}
   \caption{Comparison of estimated change point sets between algorithms and against ground truth for the Bernoulli
    family using the Jaccard Index (top) and the number of change-points (bottom). The true number of change points $k^*$ is indicated on top of each figure and marked on the purple dotted horizontal line.}
    \label{fig:metric_bernoulli}
\end{figure}

For both distributions, we varied the number of samples $n$ from $50$ to $500$ in steps of $50$. For each sample size, we simulated $1000$ data sets with $n$ samples. We fixed the number of variables as $m = 200$ and the number of change points  $|C^*| = k^*$ to take values in $\{10, 20, 50\}$.
For each $k^*$, the change points were sampled without replacement from a uniform distribution in $[1, 199]$ and where maintained fixed for all data sets and all sample sizes. 
Bernoulli parameters were sampled independently from a uniform distribution. For the Gaussian distribution, means were sampled independently from a $N(0, \sqrt{5})$, and variances from an $\mbox{Exp}(1)$. These parameters also remained fixed for all data sets and all sample sizes.
The penalization constant $\lambda$ was selected from the set $\{0.1, 1, 10\}$. The models were fitted for each $\lambda$ in the set, and then BIC was used to choose the final model. The penalization functions were set to $J(n) = \log(n)$ and $\rho(r:s) = 1$.

In order to evaluate the convergence of the change point set estimated by both, the dynamic programming and hierarchical segmentation algorithms, we considered the Jaccard Index, defined by 
\begin{equation*}
J(C_1, C_2) = \frac{|C_1 \cap C_2|}{|C_1 \cup C_2|} \,.
\end{equation*}
We also took into account the convergence of the estimated number of 
change points to $k^*$. It allow us to see when the model is heavily underestimating or overestimating the number of change points.

Figure~\ref{fig:metric_bernoulli} shows that both algorithms are converging to the true change point set as sample sizes grow. The boxplots on the top are becoming increasingly narrower and closer to $1$, which indicates total similarity between the change point sets. For $k^* = 10$, most values are indeed $1$, as indicated by the collapsed boxplots from sample sizes greater or equal to $350$. We also observe that the performance of the algorithms are very similar as sample sizes grow, since the boxplots shapes become similar. The hierarchical algorithm exhibits more variability, since the boxplots are wider and have more outliers. The graphics on the bottom row of the figure, 
that display boxplots and average lines for the estimated number of change points, help us understand also the  behavior of the algorithms as sample size increases.

For the proposed bootstrap procedure of Section~\ref{sec:bootstrap}, we evaluated how likely each index was estimated as a change point  by the hierarchical algorithm on the $k^* = 10$ scenario  (Figure \ref{fig:bootstrap_singledataset_detection}). 

\begin{figure}
\centering
\includegraphics[width = 0.8\textwidth]{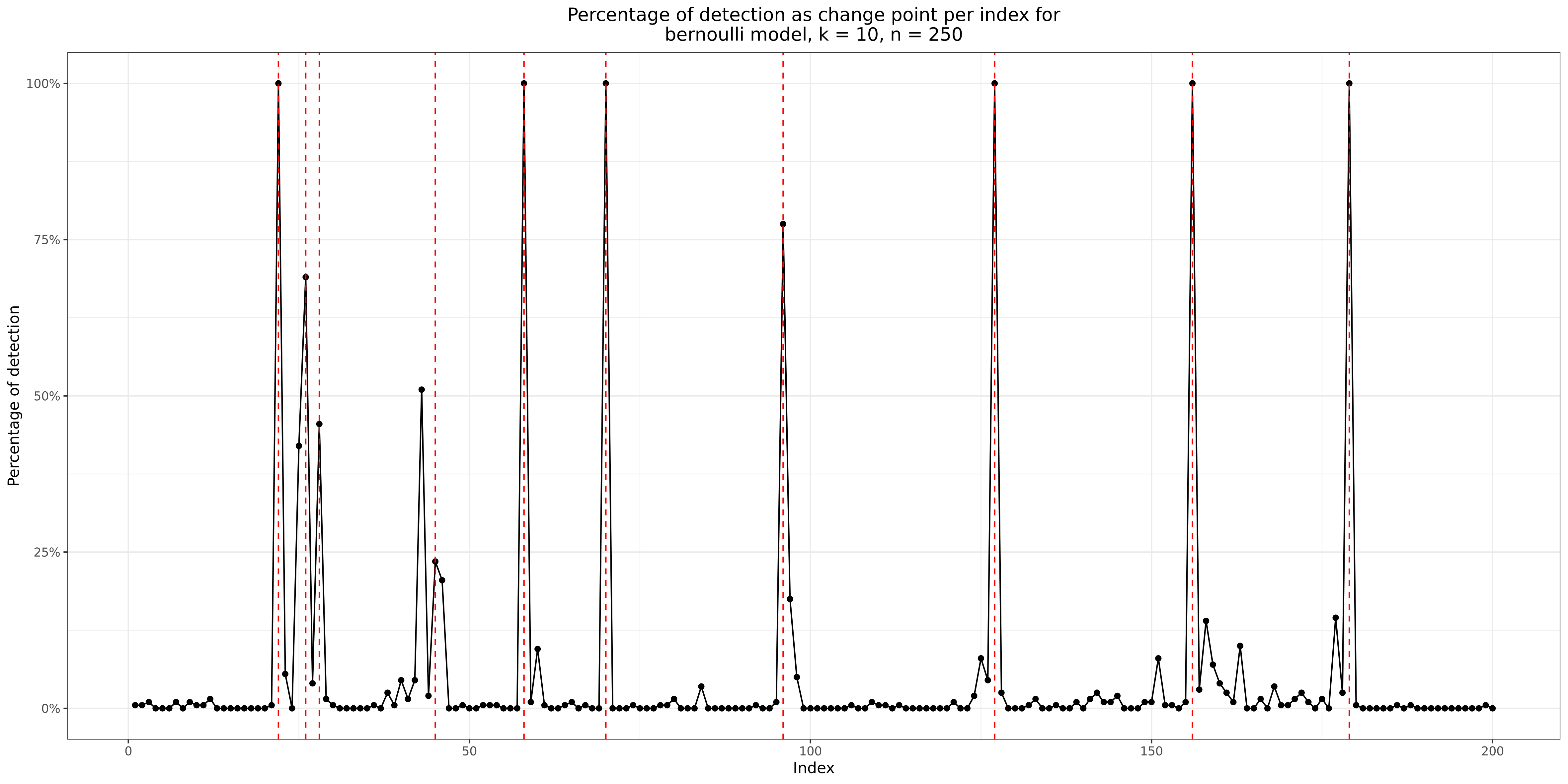}
\caption{Percentage of detection by  the hierarchical segmentation algorithm of each index as a change point. A random data set from the scenario $k^* = 10$, $n = 250$ and Bernoulli family was selected, and $200$ bootstrap samples were obtained. Black lines show how frequently the particular indexes were estimated as change points. Vertical dashed red lines indicate the position of the true change points on the simulated model.}
\label{fig:bootstrap_singledataset_detection}
\end{figure}

Intuitively, we expect the change point detection rate to be higher and with less variability if the blocks it divides have very discrepant parameters and are larger. For the Bernoulli model, we can asses how the detection rate varies with the absolute difference between consecutive parameters on each block, as observed in Figure~\ref{fig:bootstrap_avg_detection_true_cp}.

\begin{figure}
\centering
\includegraphics[width = 0.49\textwidth]{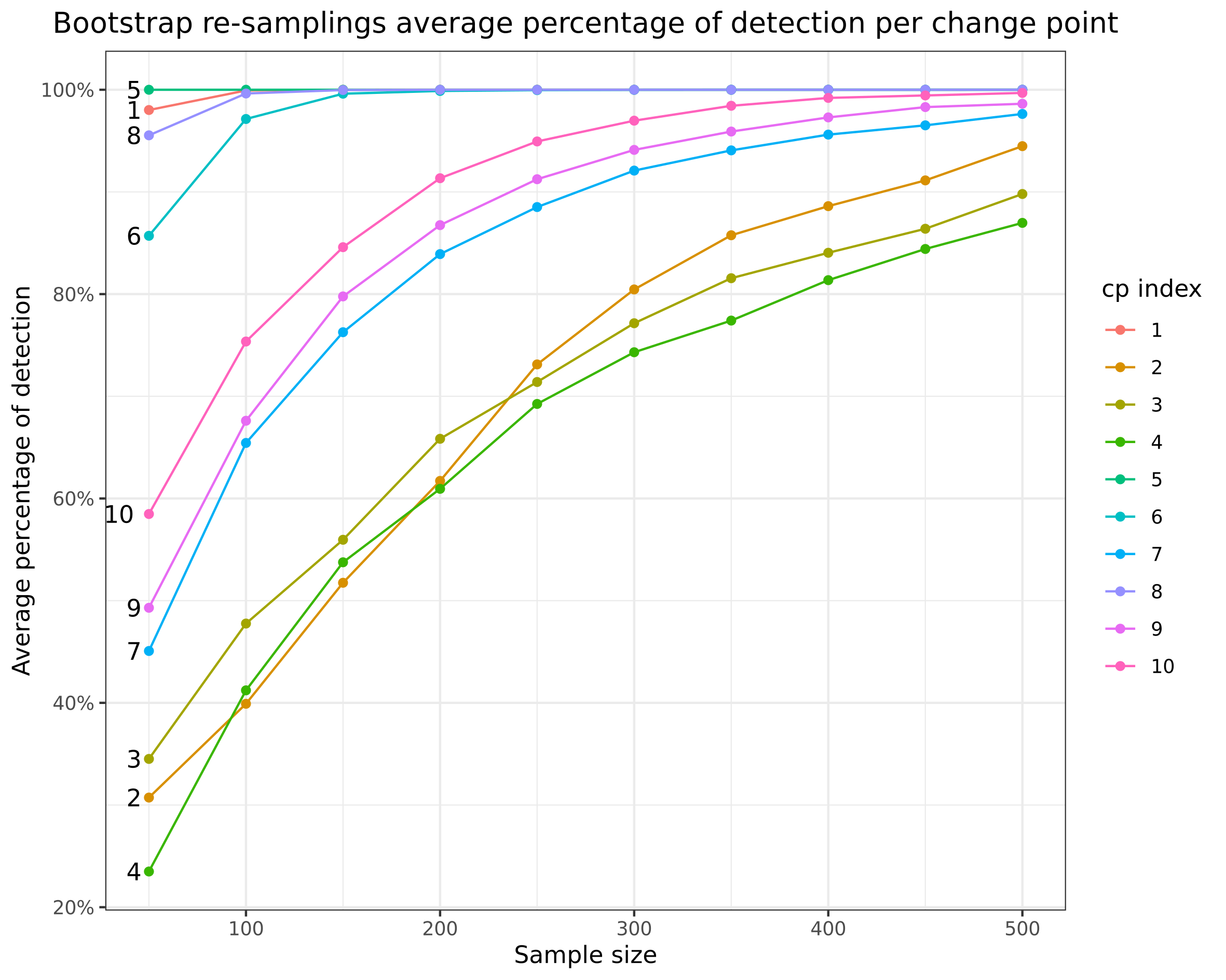}
\includegraphics[width = 0.49\textwidth]{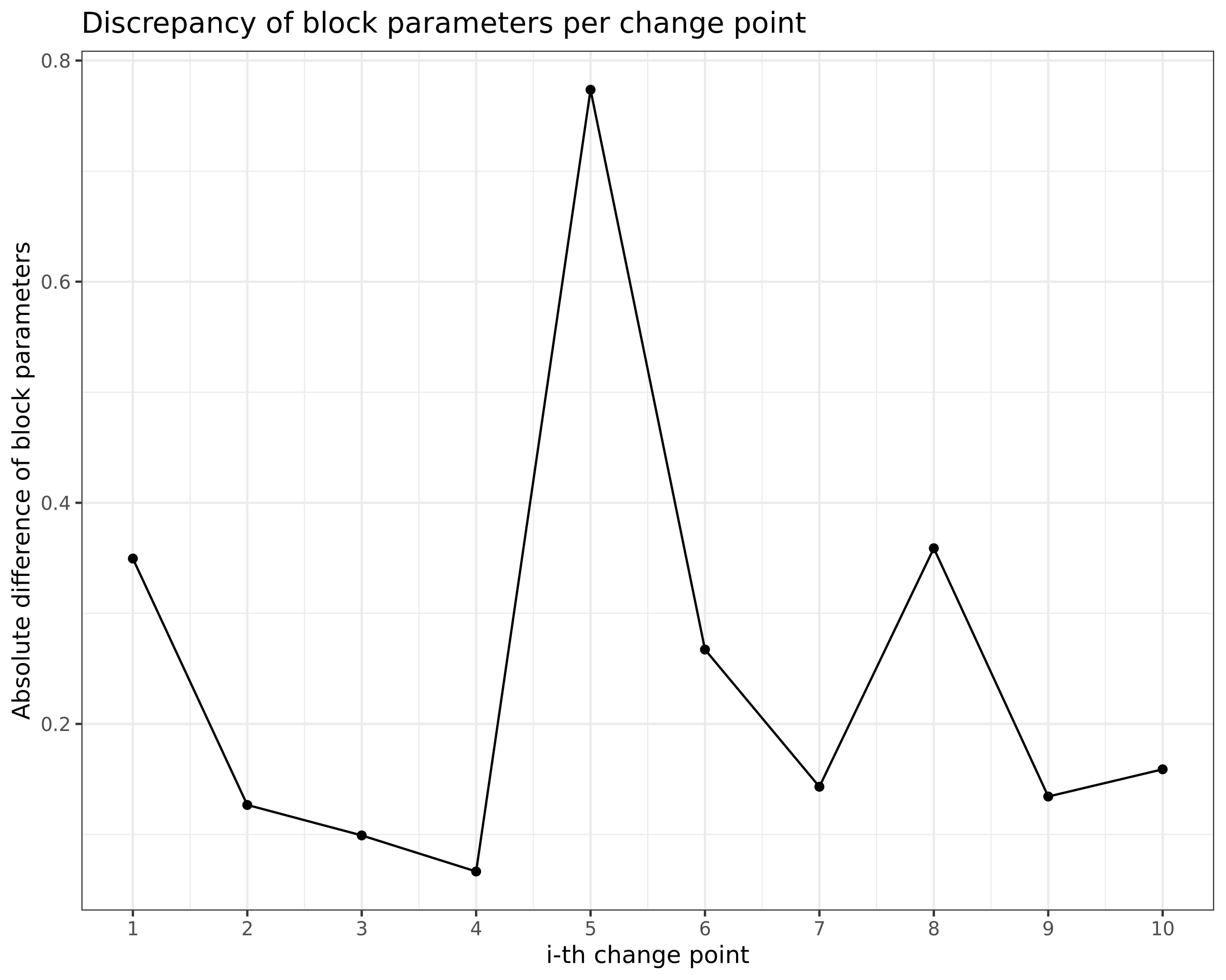}
\caption{The curves on the left figure shows the bootstrap averaged percentage of detection of true change-points as sample size grows. The average is taken over data sets with the same size. The curve on the right shows the absolute difference between the parameters of the blocks they divide. The difference is plotted with respect to the change-point index (cp index) and not location on the random vector.}
\label{fig:bootstrap_avg_detection_true_cp}
\end{figure}

The detection of each of the change points increases with sample size, and that change points which separates blocks with higher absolute difference of probability parameters are detected more frequently.

\section{ROH islands on African and European populations}



\begin{figure}
\centering
\includegraphics[scale = 0.35]{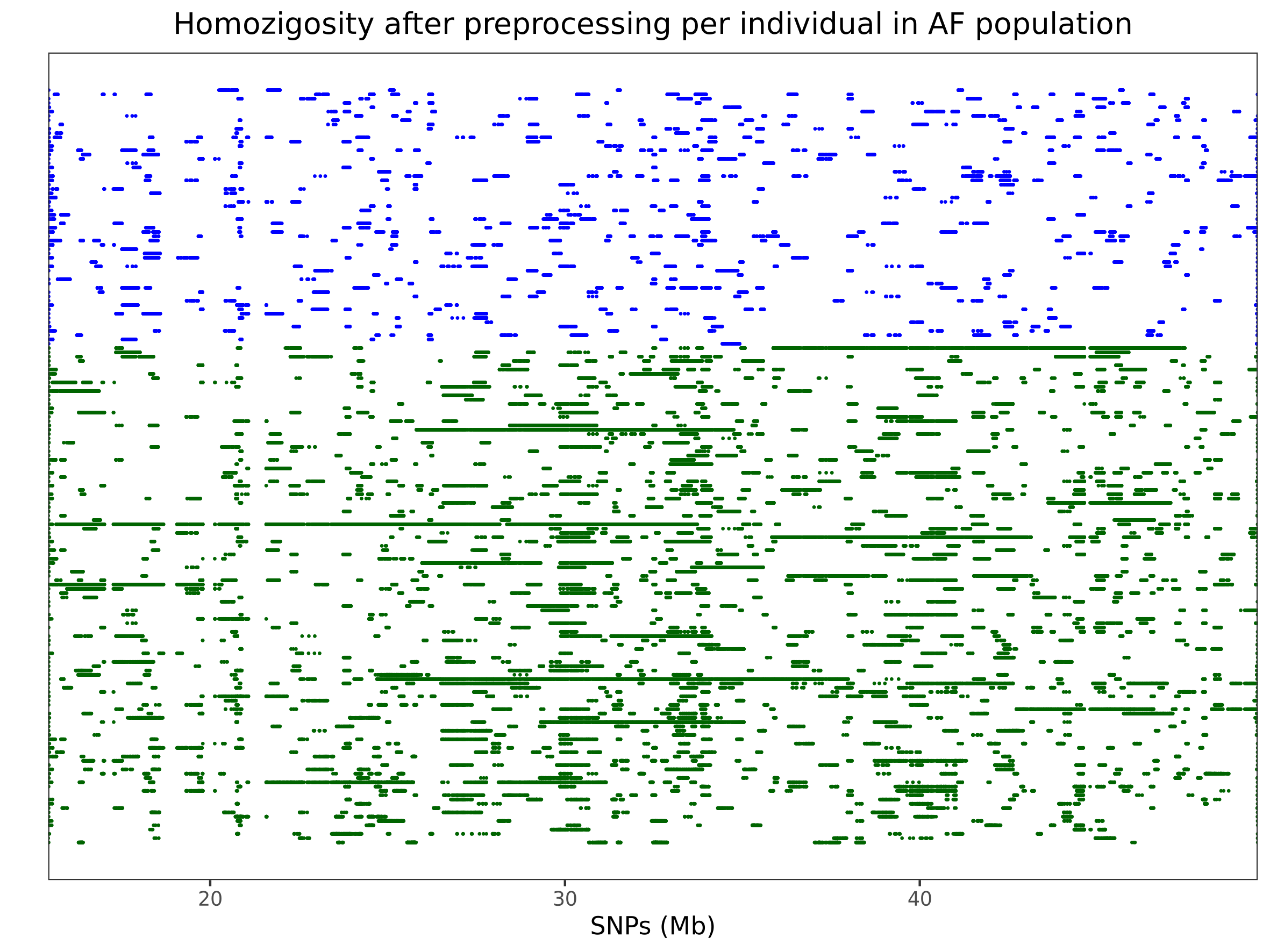}
\includegraphics[scale = 0.35]{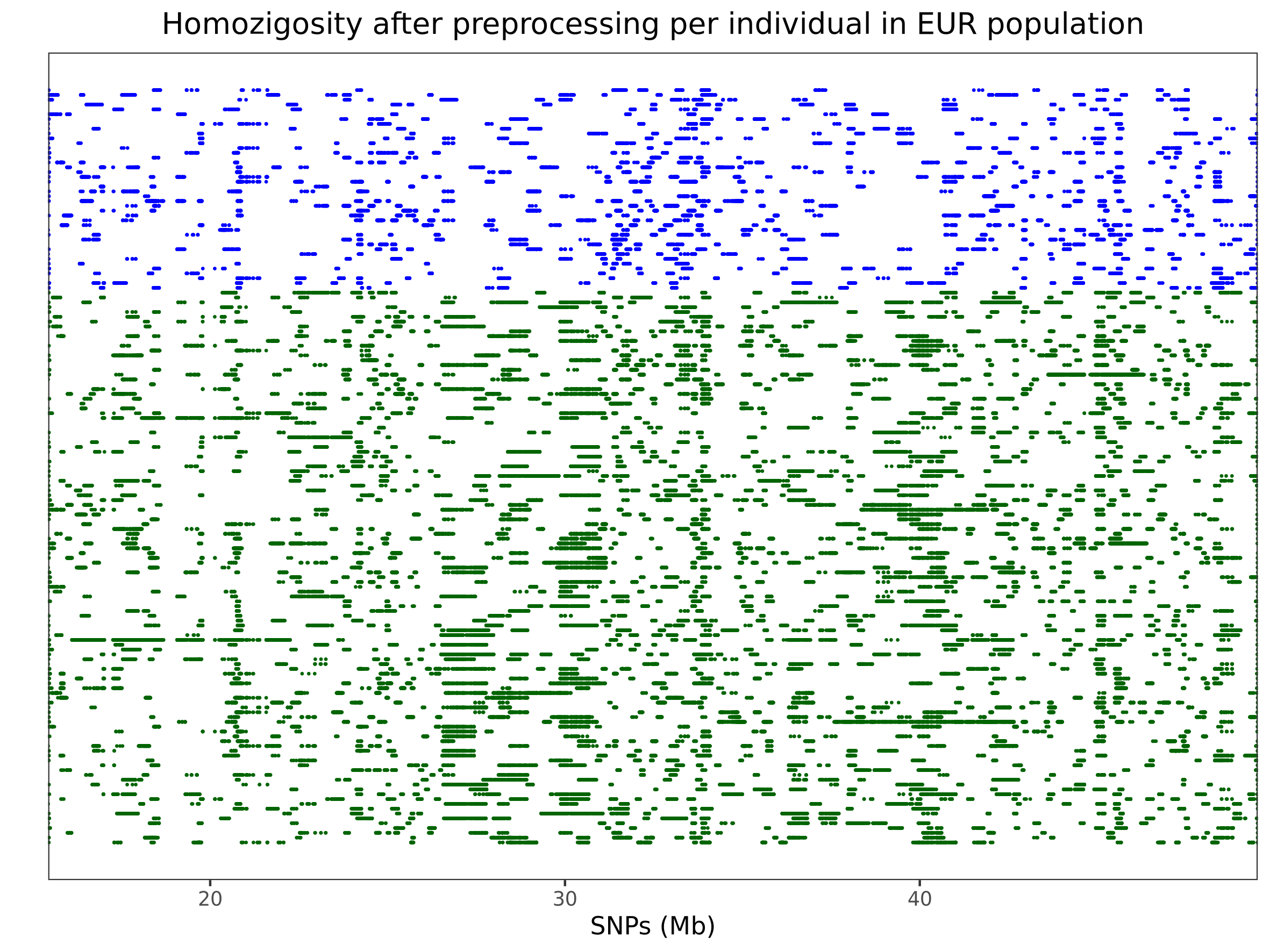}\\
\includegraphics[scale = 0.35]{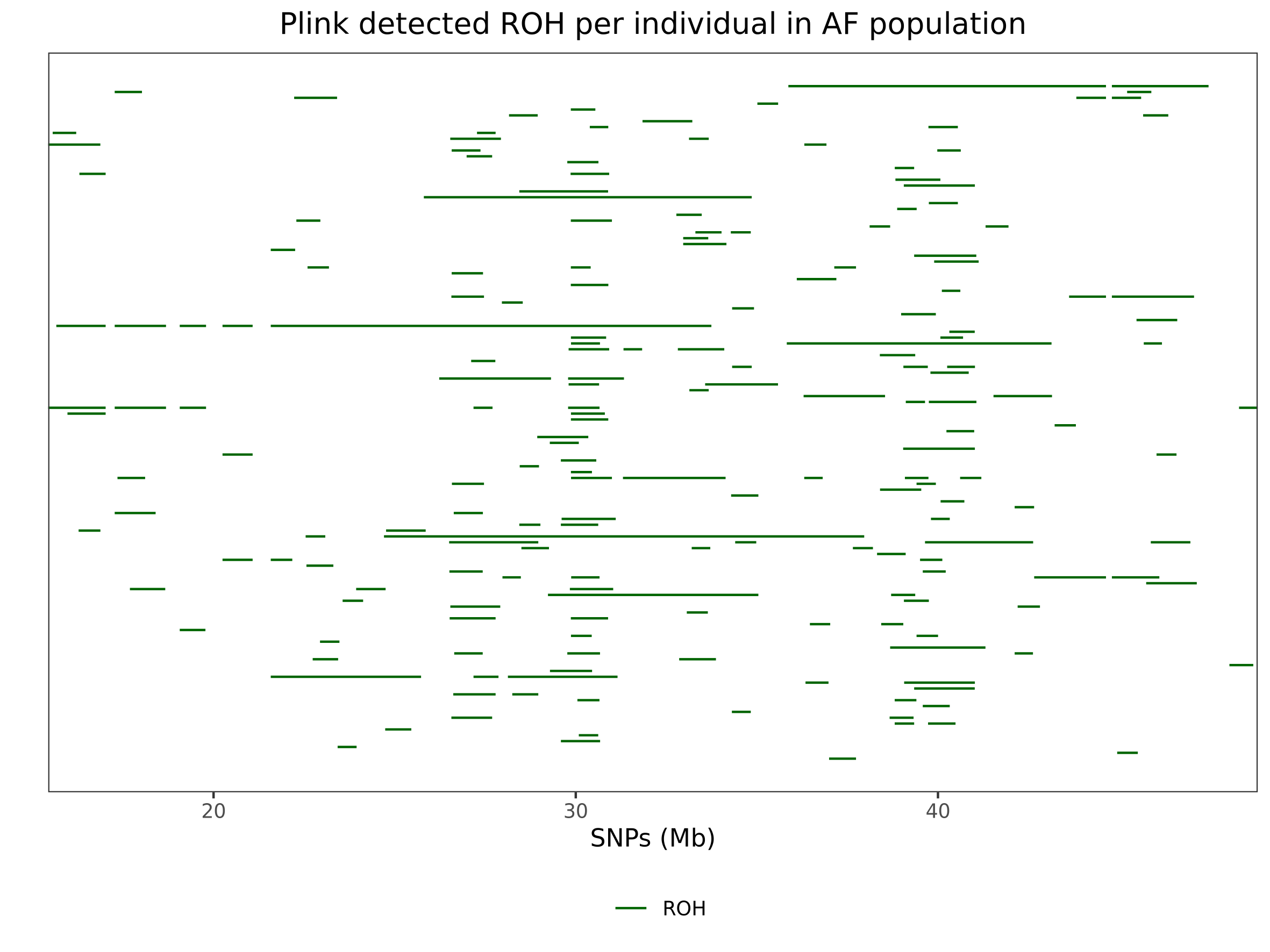}
\includegraphics[scale = 0.35]{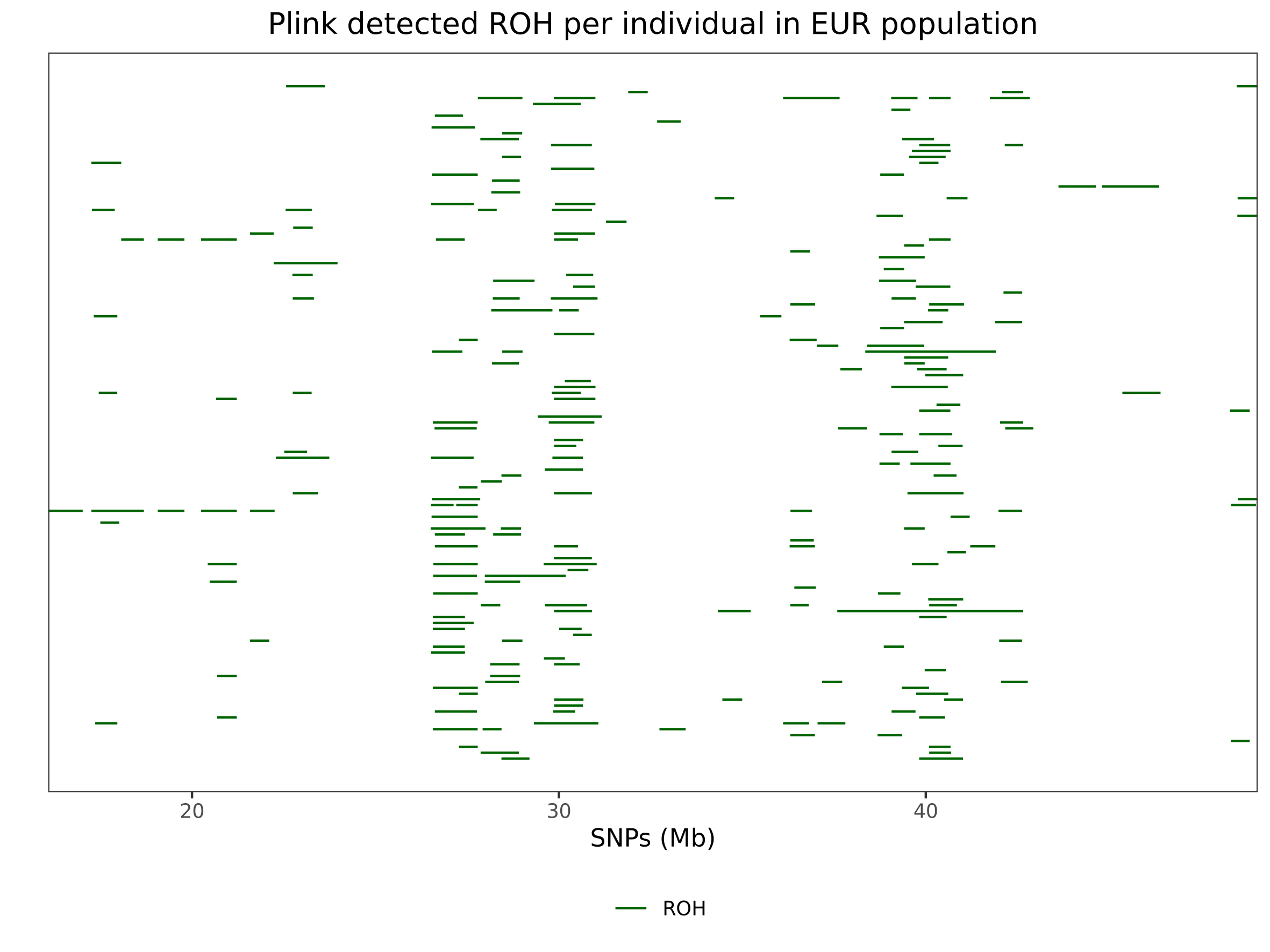}
\caption{ROH regions (green and blue segments) as defined by our method (top) and by PLINk 
(bottom) for African (left) and European (right) populations.
Each row corresponds to an individual and each column to a SNP marker. The blue segments on the top graphics correspond to individuals not having ROH detected by PLINK.}
\label{fig: AF ROH}
\end{figure}

As described in the introduction, we propose to frame the problem of ROH islands detection as a  change point detection problem assuming a Bernoulli marginal distribution for each codified SNP. Observe that we do not need to assume independence between different SNPs in order to have consistent estimators of the change-points, we only assume consecutive parameters on the blocks of the distribution are different. In particular, the ROH islands can be determined as those blocks with a high value of the estimated parameter. 




Choosing a regularization function that suits the problem is not an easy task. We can use domain knowledge to construct a proper regularization function. The first consideration is that the distance between SNPs is not uniform. That is, the distance between the $i$-th and $j$-th SNP is not $|i-j|$, but rather $|B(i)-B(j)|$, where $B$ is a function that maps each SNP to its physical location on the chromosome. The physical location of a SNP is measured as the number of base pairs before that particular SNP. The second observation is that very small blocks are usually not interesting for the analyst. It is usual to set a minimum block size in which SNPs are grouped.

Considering these observations, we define the regularization function $\rho$ for the block $r:s$ as 
\begin{equation*}
    \rho(r, s) = 
        \begin{cases}
            +\infty &\quad\text{ if}\quad \frac{|B(s) - B(r)|}{\beta} \leq T\\
            \frac{1}{|B(s)-B(r)|/\beta} &\quad\text{ otherwise.}
         \end{cases}
\end{equation*}
In the expression above, $T$  denotes a threshold for the minimal  physical distance allowed for an ROH island, and $\beta = 10^6$ is a scaling factor to work on mega bases unit. 
The regularization function $R(C)$ in \eqref{pl} is then defined as the sum of the function $\rho$ over the different blocks in $C$, as in \eqref{Radd}.  


The SNPs data we analyzed was obtained from the  Human Genome Diversity Project (HGDP), and consists of 
approximately 600,000 SNP markers from Illumina HuHap 650k platform \citep{li2008worldwide}. We considered
individuals from 
African and 
European populations.
On this dataset, each row represents an individual from the population, and each column corresponds to each SNP.

For each population, we performed ROH identification for each individual with the 
criteria described by \citet{mcquillan2008runs} and \citet{kirin2010genomic}, using the software PLINK v1.9 \citep{purcell2007plink}, and compared to the data after preprocessing. The results for chromosome 22 on both populations are shown in Figure~\ref{fig: AF ROH}. 

We then estimated the change-point set and parameters on each block for each population using a value of $\lambda$ that was selected using BIC from the set $\{0.1, 1, 10\}$. The threshold size $T$ was set to be $1\%$ of the chromosome size, and the penalization for the sample size set to $J(n) = \sqrt{n}$. Figure~\ref{fig: ROH_comparison}  shows a comparison between the high frequency regions  detected by PLINK and by our method.

For both populations, we observe that the peaks with higher ROH  frequency detected by PLINK in general correspond to blocks with higher parameter detected by the change point detection method. 
There are also blocks with high parameter not detected by PLINK, showing that the methods do not always provide the same ROH islands. This can be understood if we look back at Figure~\ref{fig: AF ROH}. There are regions that have a very high number of overlapping blocks in the processed data, but that do not appear at PLINK. Since they correspond to regions with uninterrupted homozigosity SNPs in the processed data, the output of our method will detect such regions as a block with very high probability parameter, but not in PLINK, since almost no individuals with ROH have detected ROH at that region. 
These contrasting regions can rise due to small values of the preprocessing parameter, PLINK's criteria, or the presence of a high amount of missing data in a specific genomic region.

Finally, Figure~\ref{fig: ROH_bootstrap} shows an application of the bootstrap to provide how reliable is the estimate of the change point set. We see some spikes on the percentage of detection near the start of most blocks, indicating these are indeed the most detected indices.
However, notice that the spikes do not have a very high percentage of detection. This is due to the fact that the change points oscillate near the ends of the segments. Since only one change point can be detected at a time in such regions due to the restriction imposed on the minimum block size, the confidence of detection  actually refers to the confidence of the change points on a given region and not on a specific position. A last interesting observation is that some blocks have almost no percentage of change points detected in their central parts, indicating that the model is confident about the boundaries of most of the blocks. 

\begin{figure}
\centering
\includegraphics[scale = 0.5]{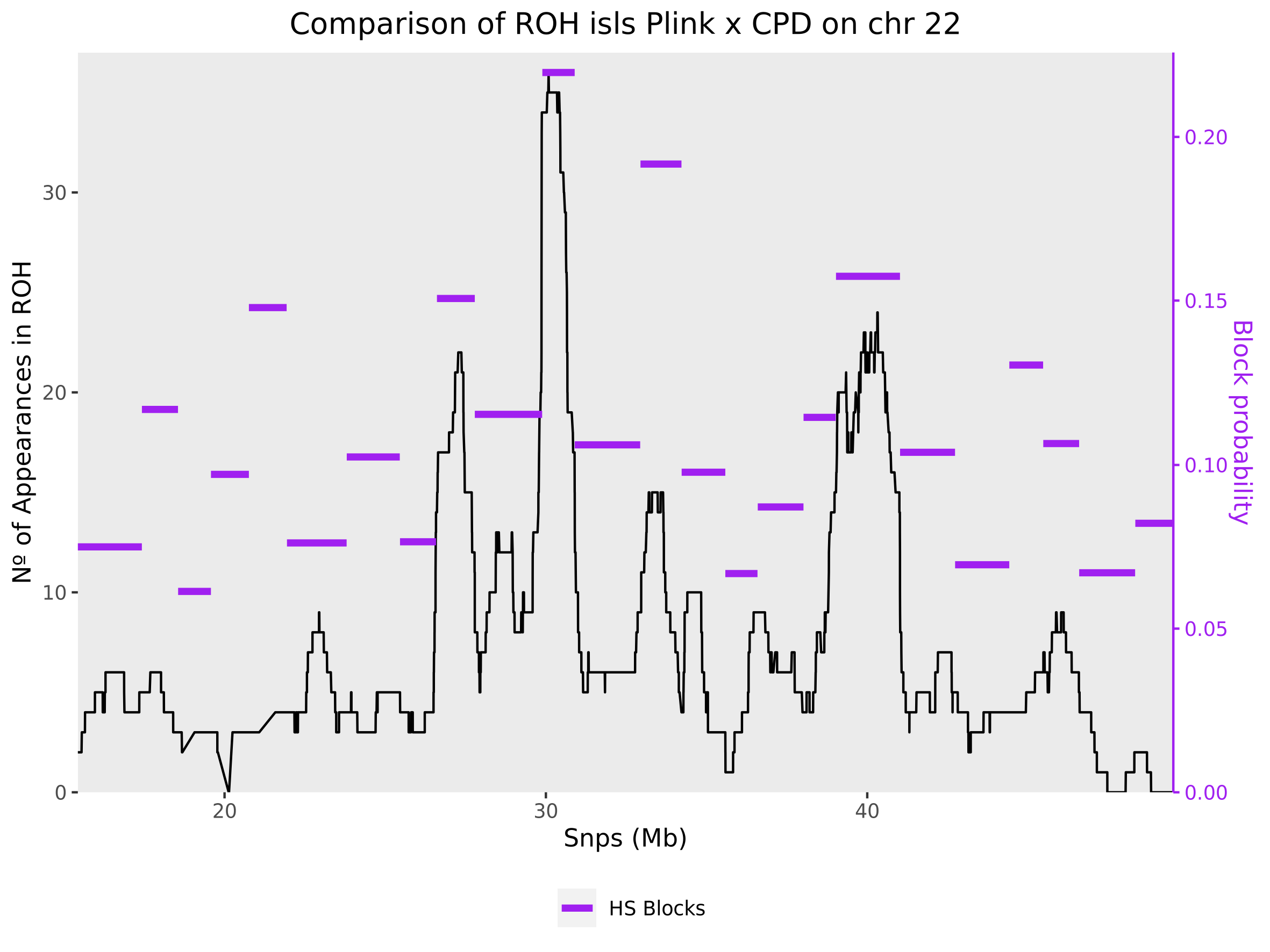}

\includegraphics[scale = 0.5]{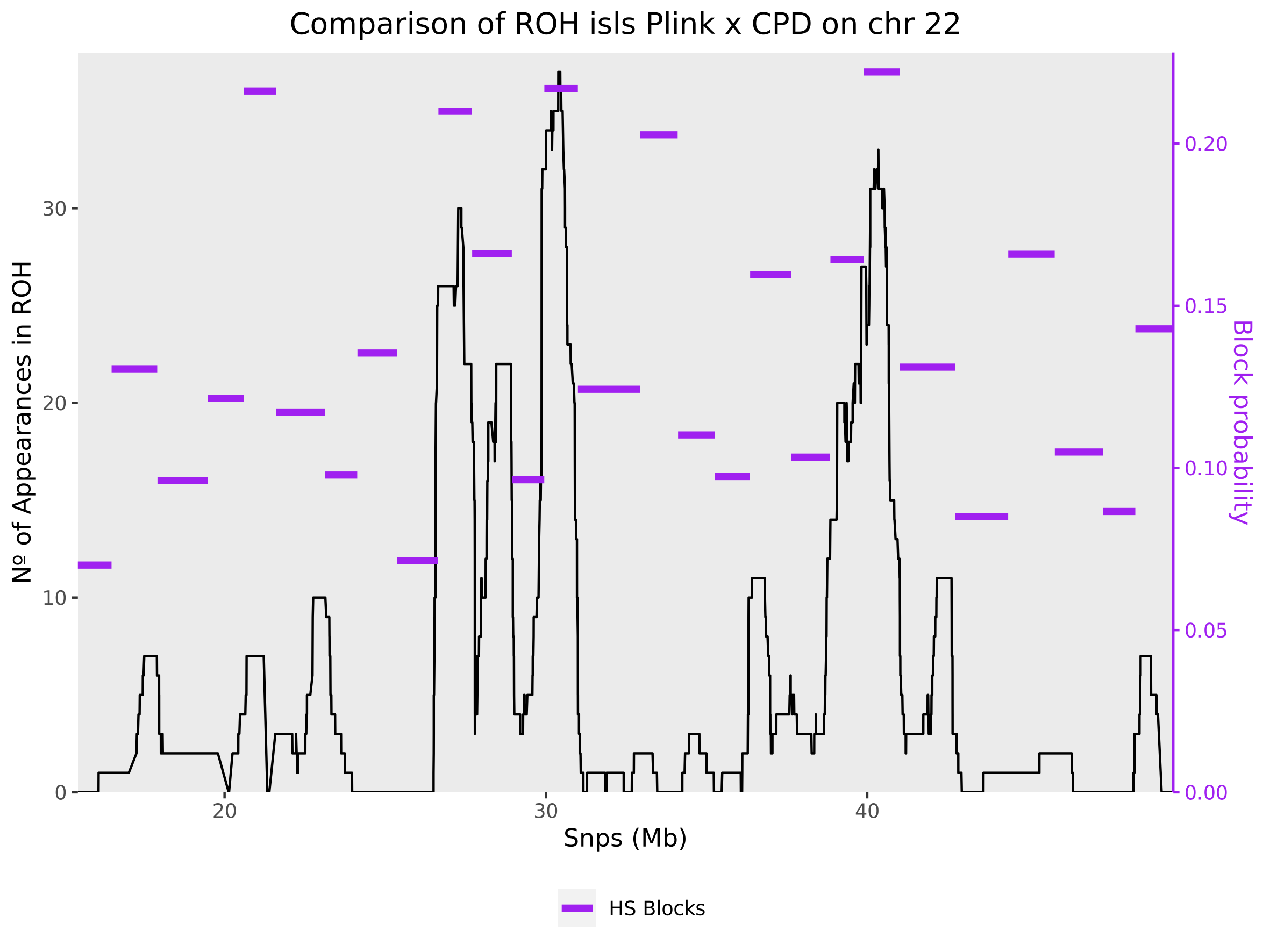}
\caption{Comparison between the blocks detected by the change-point method and the methodology proposed in \citet{mcquillan2008runs} and \citet{kirin2010genomic}, using PLINK. The black lines indicate the frequency of 
occurrence of each SNP in an ROH as detected by PLINK (left scale). The purple lines indicate the blocks
detected by our methodology. The height of each block corresponds to the estimated parameter on each block (right scale).}
\label{fig: ROH_comparison}
\end{figure}

\begin{figure}
\centering
\includegraphics[scale = 0.5]{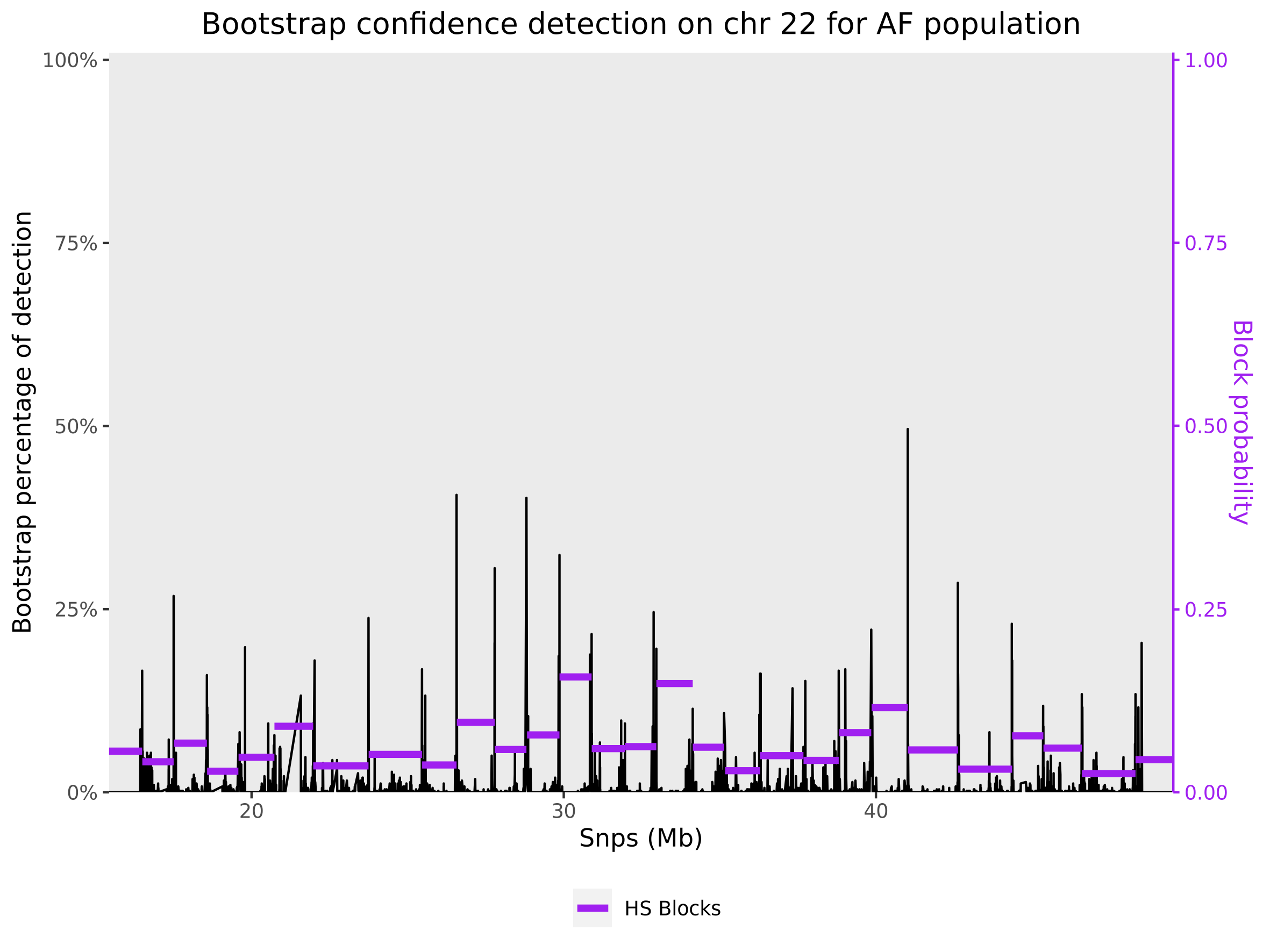}
\includegraphics[scale = 0.5]{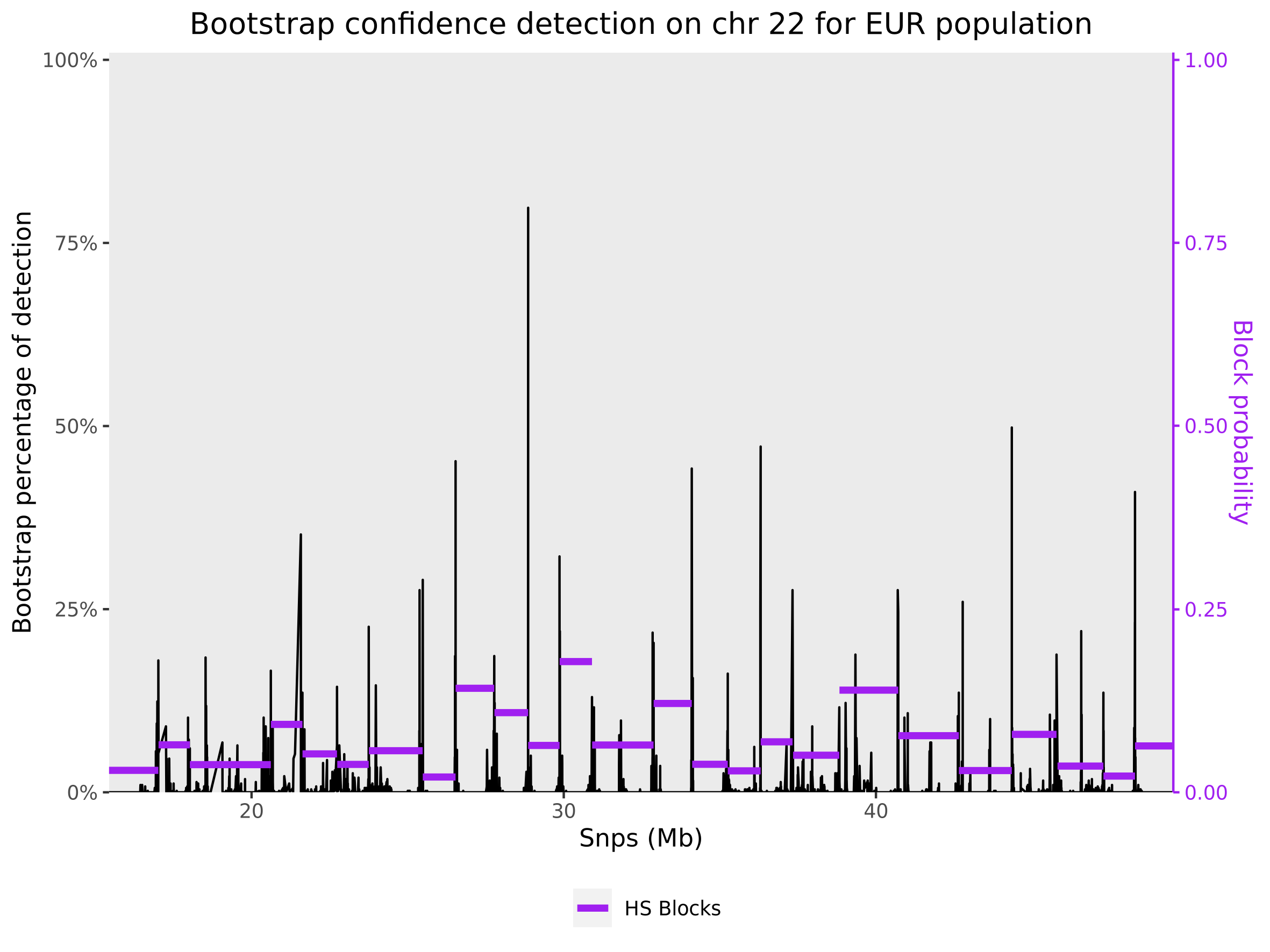}
\caption{Bootstrap confidence of detection of the estimated change point set with 500 replications. Left scale shows the bootstrap percentage of detection as a change point per index. The purple lines indicate the blocks
detected by our methodology. The height of each block corresponds to the estimated parameter on each block (right scale).}
\label{fig: ROH_bootstrap}
\end{figure}

\section{Discussion}

In this paper, we studied a different change point detection scenario that arises with new applications. The inferential goals are the same as the classical problem, but consistency is studied differently, and exact estimation for the change point set is required. We proved that a penalized maximum likelihood approach can be applied and that this method consistently estimates the change point set for different classes of distribution families. We also showed how the bootstrap could be used to obtain confidence estimates for each change point and the whole change point set.

The penalization constant $\lambda$  appearing in the definition of the estimator has to be set by the user or estimated from data. Our theorems guarantee that the algorithms are consistent as long as the constant is fixed. However, since the speed of convergence depends on this constant, it is interesting to select a value that enhances convergence speed.
    
If our method is used in a regression setting as part of feature selection for other algorithm, cross-validation can be used to select a proper value for $\lambda$. If it is used in a more descriptive manner, BIC can be used for selection, such as applied in the simulations. As long as the proposed set for lambda is finite and fixed, the algorithms continue to be consistent. 
The simulation study suggests that the choice of the algorithm should depend on sample sizes. The performance of the hierarchical algorithm seems to be equivalent to the exact solution provided by the dynamical programming algorithm for bigger sample sizes. The exact algorithm remains better for small sample sizes since it seems more robust when selecting the penalization constant. There is a clear speed gain using the greedy approach since it is asymptotically $O(m(n + |C|))$, while the exact algorithm has complexity $O(nm + m^2)$. This difference is crucial in high dimensional analysis, such as applications in genetics. Of course, limiting the number of estimated blocks also aids in time speed and might avoid over-segmentation.

This new approach for change point detection is motivated by the problem of identifying homozygosity islands on the genome of individuals in a population. Our method directly tackles the issue of determining the homozygosity islands at the population level without analyzing single individuals and then combining the results, as is made nowadays in state-of-the-art approaches. Applying this method to real data of two populations from the Human Genome Diversity Project (HGDP) showed the potentiality of these algorithms to highlight highly homozygous regions in the genome.  The non-homogeneity of the regularization function $R$ can provide flexibility to incorporate more domain knowledge of the application area.

There is much to explore in future research. A first step will be to check if the assumptions hold for a wider class of distributions, such as  Exponential Family, finite-state Markov Chains and Multivariate Gaussian within each block. For the i.i.d Gaussian with unknown mean and variance, the proof that assumption \ref{ass_PL} holds is very similar to the one provided for the case with known variance in the supplementary material. However, proving assumption \ref{ass_HS} seems to be more challenging. A second interesting question is to prove or give a counter-example of whether \ref{ass_HS} implies \ref{ass_PL}. That is if the convergence of the greedy algorithm implies the convergence of the exact algorithm as well. A third interesting question is how to derive rates of convergence of the change point set estimators.


\section{Additional information}
Simulations and computation of the estimators were performed using the R software \cite{Rsoftware} and the
R package \href{https://github.com/Lucas-Prates/blockcpd}{blockcpd} developed
for this task. The package provides easy-to-usedr  of usage and implements the algorithms for the Bernoulli, Gaussian with unknown mean and variance, Exponential and Poisson families for the i.i.d case. An implementation for a two state Markov Chain is also provided. The c++ R framework provided by the Rcpp package \citep{rcpp_2011,rcpp_2017} was used to overcome performance bottlenecks.

    

\acks{ }

This article was produced as part of the activities of FAPESP's\footnote{S\~ao Paulo Research Foundation, Brazil}  Research, Innovation and Dissemination Center for Neuromathematics, grant 2013/07699-0, and FAPESP's projects ``Stochastic modeling of interacting systems'', grant 17/10555-0  and  ``Model selection in high dimensions: theoretical properties and applications'',  grant 2019/17734-3. During the development of this work L.P was
supported by a CAPES\footnote{Coordination of Superior Level Staff Improvement, Brazil} fellowship.
F.L is partially supported by a CNPq research fellowship, grant 311763/2020-0.

\appendix

\section{Proofs of Theoretical Results}
\label{appendix:proof_of_consistency}

In this section we present the proofs of Proposition~\ref{prop: hier_complexity} stated in Section~\ref{sec:computation} and all the theoretical results stated in Section~\ref{sec:consistency}.

\begin{proof}[Proof of Proposition \ref{prop: hier_complexity}]
\label{proof: hier_complexity}

We will first prove that the algorithm will eventually almost surely do exactly $2k_{C^*} - 1$ recursive calls by induction on $k_{C^*}$.  Consider the base case $C^* = \{0,m\}$, when there is no change point except the extremes. The first run of the algorithm is always on $1:m$. Since it is asymptotically correct, it halts after the first iteration. Therefore, the number of recursive calls is $1 = 2k_{C^*} - 1$.
Now suppose that the formula holds for $k_{C^*} \leq K-1$ for any value of $m$, we will prove that it holds for $k_{C^*} = K$. At the first call, the algorithm selects a change point $c$ and creates a recursive call on $1:c$ and on $(c+1):m$. Let $N_l$ be the number of change points on $1:c$ and $N_r$ be the number of change points on $(c+1):m$ Using induction hypothesis, we have that the algorithm does $2(N_l + 1) - 1$ calls on $1:c$ and $2(N_r + 1) - 1$ on $(c+1):m$. Notice that $N_l + N_r = K - 2$ because $k_{C^*} = K$ implies there are $K-1$ change points for the original change point set, but $c$ is not counted on $N_l$ or $N_r$. Hence, the total number of calls is
\[
1 + (2N_l + 1) + (2N_r + 1) =  2(N_l + N_r) + 3 = 2K - 1\,. 
\]
To obtain the final complexity, note that a call in the interval $r:s$ the algorithm does $|s - r + 1| \leq m + 1$ comparisons and memory checks, and therefore has complexity of $O(mk_{C^*}))$. 
\end{proof}

\begin{proof}[Proof of Theorem~\ref{PL_Consistency}]
First we will prove that $\widehat{C}(\bx^n)$ will
almost surely contain $C^*$. Let $C\in\C$ such that 
$C\not\supseteq C^*$. Then by (PL1) in Assumption~\ref{ass_PL} we have that eventually almost surely 
\[
\frac1nl(C;\bx^n) - \frac1nl(C^*;\bx^n) \;<\; -\frac{\alpha}2 \,.
\]
On the other hand, as $J(n)=o(n)$ and $R(C)$ is bounded  we have that
\[
\lambda (R(C) - R(C^*)) \frac{J(n)}n \;\to \; 0
\]
as $n\to\infty$. Therefore, eventually almost surely we have
\[
\pl(C^*;\bx^n) \;< \; \pl(C;\bx^n)\,.
\]
As the number of $C\not\supseteq C^*$ is finite we have that eventually almost surely $\widehat C\supseteq C^*$. 
Lets prove now that $\widehat C \not\supset C^*$. Assume $C\supset C^*$ so that $R(C) > R(C^*)$ and $l^*(C) = l^*(C^*)$. Then, by (PL2) we have that 
\begin{equation}
\begin{split}
\pl(C^*;\bx^n) - \pl(C;\bx^n)\;&\leq\; v(n) + \lambda(R(C^*) - R(C)) J(n)\\
& <\; 0
\end{split}
\end{equation}
eventually almost surely as $n\to\infty$. As a result we obtain that $\widehat C(\bx^n)=C^*$ eventually almost surely as $n\to\infty$. 
\end{proof}

    
We now present the proof of the consistency of the hierarchical estimator.
	
\begin{proof}[Proof of Theorem~\ref{HS_Consistency}]
The proof begins by showing that, at any given possible scenario, the algorithm takes a correct choice almost surely. Then, an inductive argument will guarantee that the algorithm is consistent. 
For every integer interval $I = r:s$, the possible scenarios are:
\begin{itemize}
\item [(a)] There are no change points in $I\setminus\{s\}$;
\item [(b)] There are change points in  $I\setminus\{s\}$.
\end{itemize}
The correct decision for the algorithm in case (a) is to halt, not performing more recursive calls for that interval. For  case (b), the correct decision is to choose any of the change points available and perform recursive calls on the sub intervals. We show that the algorithm will take the correct decision almost surely for both cases.\\
For case (a), suppose $I$ has no change points inside and let $u \, \in\, r:(s-1)$. 
By (H2) in Assumption~\ref{ass_HS}  we have that 
\begin{align*}
h_I(s) - h_I(u) \;&=\;  -\left[l(I; \mathbf{x}_I^n) - l(r:u; \mathbf{x}_I^n) - l((u+1):s; \mathbf{x}_I^n)\right] + \\
&\quad\; + J(n)\lambda\left[\rho(r, s) - \rho(r, u) -\rho(u+1, s)\right]\\
& <\; v(n) + J(n)\lambda\left[\rho(r, s) - \rho(r, u) -\rho(u+1, s)\right] \\
&< \;0\,,
\end{align*}
eventually almost surely as $n \to \infty$, because $\rho(r, s) - \rho(r, u) -\rho(u+1, s) < 0$. Hence, no splitting will be done eventually almost surely, and the algorithm will not perform more recursive calls inside $I$.\\
For case (b) we have to prove that the algorithm will almost surely choose true change points to split the interval. The inequality (H1) on Assumption~\ref{ass_HS} implies that, for any $u$ that is not a change point, there exists a change point $c^*$ in $I\setminus\{s\}$ such that
\begin{equation*}
h_I^*(c^*)  \; < \;h_I^*(u)\,.    
\end{equation*}    
Then as $J(n) = o(n)$  we have  that eventually almost surely as $n\to\infty$
\begin{equation*}
 \frac1n h(u) - \frac1n h(c^*) \;\geq\; \frac12(h^*(u) - h^*(c^*)) \;>\; 0 \,,
\end{equation*}    
and therefore $\hat c \in C^*$, eventually almost surely as $n\to\infty$.


We finish the proof by using mathematical induction on the number of variables $m$. If $m = 1$, then there are no change points on the model, and the algorithm will not even have comparisons to make. Hence, the estimated change point set will be empty by construction and therefore equal to the true set of change points.
Assume now that the algorithm is consistent for every vector of dimension $\tilde{m} <= m-1$. We will prove that it will be consistent for vectors of dimension $m$.
The first run of the algorithm is on the interval $1:m$. By case (a),  if there are no change points on this interval, the algorithm will almost surely not split the interval and will halt. Hence, the estimated set will almost surely be equal to the true set of change points. On the other hand, if there are change points on the interval $1:m$, the proof in case (b) shows us again that the algorithm eventually almost surely takes the correct choice and splits the interval at a change point $c\in C^*$. After the split, recursive calls are made on $1:c$ and on $(c+1):m$. But the length of these vectors is at most $m-1$, and by the induction hypothesis, the algorithm will eventually almost surely retrieve all the change points in $1:c$ and $(c+1):m$ exactly. By a union of those change point sets,  we have that the final estimated change point set is equal to $C^*$ eventually almost surely as $n\to\infty$.
\end{proof}

\begin{proof}[Proof of Lemma~\ref{conditions_ineq_HS}]
Let $h_I^*(u)$ be the minimum value of $h_I^*$ on $I$. First, notice that by (a),  $h_I^*(u) \leq h_I^*(c) < h_I^*(s)$, so the minimum is strictly smaller than the value of the function at the end of the interval. 
We now show that $u$ must be a change point in $C^*$.  Suppose that $u\not\in C^*$, and let $I_j^*=[c_{j-1}^*,c_{j}^*]$ be the  unique interval that contains $u$. Since $h_I^*$ is concave in $[c_{j-1}^*, c_{j+1}^*]$, if the minimum is attained at an interior point, then $h_I^*$ must be constant on $I_j^*$. However, by (c), this would imply that $h_I^*(u) = h_I^*(s)$, which is a contradiction. Then $u\in C^*$ and
\begin{equation*}
 \min_{c \in I\setminus\{s\}\cap C^*} \; h_I^*(c)\;<\; \min_{c \not\in I\setminus\{s\}\cap C^*} \; h_I^*(c) \,. \hfill\qedhere
\end{equation*}
\end{proof}
    
Before presenting the proof of Lemma~\ref{conditions_iid_HS} we state and prove the following basic lemma. 

\begin{lemma}\label{basic_lemma}
Let $f,g:\mathbb{R}\to\mathbb{R}$ be twice differenciable convex functions. If there exists a constant $\alpha$ such that
\begin{equation*}
f(x) + g(x) = \alpha \,,
\end{equation*}
then $f$ and $g$ must be linear functions.
\end{lemma}    
    
\begin{proof}
Differentiating both sides twice we obtain that 
\begin{equation*}
f''(x) + g''(x) = 0 \,.
\end{equation*}
Since both $f''(x) \geq 0$ and $g''(x) \geq 0$, then $f''(x) = 0 = g''(x)$ and the result follows.
\end{proof}
    
\begin{proof}[Proof of Lemma~\ref{conditions_iid_HS}]
Observe that it is enough to consider $h^*_I$  as given by 
\begin{equation*}
h_I^*(u) = -(u-r+1)\psi(\theta_{r:u}) - (s-u)\psi(\theta_{(u+1):s}) \,,\quad u\in I\,,
\end{equation*}
as the constant $C$ does not depend on $u$. For any $u\in [c_{j-1}^*,c_j^*]$ define $t(u) = \frac{c_{j-1}^*-r+1}{u-r+1}$. Then we can write 
\begin{equation*}
\theta_{r:u} = t(u)\theta_{r:c_{j-1}^*} + (1-t(u))\theta_{j}^* \,.
\end{equation*} 
We now check each one of the condition of Lemma~\ref{conditions_iid_HS}, beginning with (b).  To prove that  $h^*_I$ is concave on the interval $[c_{j-1}^*, c_{j}^*]$, it is sufficient to show that  $(u-r+1)\psi(\theta_{r:u})$ and $(s-u)\psi(\theta_{(u+1):s})$ are convex on this interval. 
Take $g(u) = (u-r+1)\psi(\theta_{r:u})$, and treat the vectors as column vectors. Then the first derivative of $g$ is
\begin{align*}
g'(u) \;&=\; \psi(\theta_{r:u}) + (u-r+1)t'(u) (\theta_{r:c_{j-1}^*}-\theta_{j}^*)^T \nabla \psi(\theta_{r:u})\\
&=\; \psi(\theta_{r:u}) -t(u)(\theta_{r:c_{j-1}^*}-\theta_{j}^*)^T\nabla \psi(\theta_{r:u})\,,
\end{align*}
where $\nabla \psi(\theta_{r:u})$ is the gradient of $\psi(\theta)$ evaluated at $\theta_{r:u}$ and the second equality follows from the fact that $(u-s+1)t'(u) = -t(u)$.
The second derivative of $g$ is then
\begin{align*}
g''(u) 
\;&=\; \overbrace{t'(u)(\theta_{r:c_{j-1}^*}-\theta_{j}^*)^T\nabla \psi(\theta_{r:u}) - t'(u)(\theta_{r:c_{j-1}^*}-\theta_{j}^*)^T \nabla \psi(\theta_{r:u})}^{=\;0}\\
&\quad +(-t'(u)t(u))(\theta_{r:c_{j-1}^*}-\theta_{j}^*)^T H\psi(\theta_{r:u})(\theta_{r:c_{j-1}^*}-\theta_{j}^*)\\
&\geq\; 0\,,
\end{align*}
because $H\psi(\theta_{r:u})$, the Hessian matrix of $\psi$ evaluated at $\theta_{r:u}$,  is positive definite and  $-t'(u)t(u) \geq 0$.
We conclude that $g$ is convex on $[c_{j-1}^*, c_j^*]$. An analogous argument shows that $(s-u)\psi(\theta_{(u+1):s})$ is also convex, and then we finish verifying condition (b) in Lemma~\ref{conditions_ineq_HS}.
Now we will show that, if $h_I^*$ is constant on an interval $[c_{j-1}^*, c_{j}^*]$, then it must be equal to $h^*_I(s)$ on this interval.  So suppose that $h_I^*$ is constant on $[c_{j-1}^*, c^*_{j}]$. Then, for some $\alpha$ we must have that 
\begin{equation}\label{equality_alpha}
-(u-r+1)\psi(\theta_{r:u}) - (s-u)\psi(\theta_{(u+1):s}) \;=\; \alpha 
\end{equation}
for all $u\in [c_{j-1}^*, c^*_{j}]$. Lemma~\ref{basic_lemma} implies that $(u-r+1)\psi(\theta_{r:u})$ and $(s-u)\psi(\theta_{(u+1):s})$ must be linear
functions on $[c^*_{j-1},c^*_j]$, therefore $\psi(\theta_{r:u})$ and $\psi(\theta_{(u+1):s})$ are constants. 
Since $\psi$ is strictly convex and $\theta_{r:u}$ (respectively $\theta_{(u+1):s}$) is a convex combination of $\theta_{r:c_{j-1}^*}$ and $\theta_j^*$ (respectively of $\theta_{c_{j}^*:s}$ and $\theta_j^*$), we conclude that 
$\theta_{r:u} = \theta_{j}^* = \theta_{(u+1):s}$ for all
$u\in[c^*_{j-1},c_j^*]$. 
Moreover we have that   
\begin{equation*}
\theta_I \;=\; \frac{1}{|I|}\bigr[ (u - r + 1)\theta_{r:u} + 
(s - u)\theta_{(u+1):s}\bigl] \;=\; \theta_j^*
\end{equation*}
and therefore $h_I^*(u) = h_I^*(s)$, implying condition (c) in Lemma~\ref{conditions_ineq_HS}.
We finish the proof by showing condition (a) in Lemma~\ref{conditions_ineq_HS}, that is 
that the minimum is attained at the interior of $I$. Suppose that this is not the case, that is that the minimum is attained at $h_I^*(s)$. Since $h_I^*(s)$ is also a maximum, because by definition
\[
\frac1n l(r:u) + \frac1n l((u+1):s) \;\leq\; \frac1n  l(r:s) \quad\text{for all }u\in I
\]
we have that the function in the whole interval $I$ must be constant. Following the same arguments as in the proof of condition (b) we conclude that $\theta_i^* = \theta_j^*$ for all $i, j \, \in \, \{1, \dots, k^*\}$, which is a contradiction with the hypothesis that $\theta_j^*\neq \theta_{j+1}^*$ for all $j=1, \dots, k^*-1$. 
\end{proof}    


We will now verify that Assumptions \ref{ass_PL} and \ref{ass_HS} hold for the families of distributions in Examples \ref{example_categorical} and \ref{example_gaussian}. These results were stated in Propositions~\ref{prop_categorical} and \ref{prop_gaussian}, respectively. 
    

\begin{proof}[Proof of Proposition \ref{prop_categorical}]
Denote by $C^* = \{c_0^*, \ldots, c_{k_{C^*}}^*\}$ and $p^* = (p_1^*, \ldots, p_{k_{C^*}}^*)$ the true change point set and parameters, respectively. We will denote by $C = \{c_0, \ldots, c_{k_C}\}$ and $p= (p_1, \ldots, p_{k_C})$ any arbitrary change point set 
and we will denote by  $I_j$, respectively $I^*_j$,  the interval between two change points in $C$, respectively in $C^*$, that is $I_j=(c_{j-1}+1):c_j$ and $I^*_j=(c_{j-1}^*+1):c_j^*$. 
We begin by verifying hypotheses (PL1) and (PL2) in Assumption~\ref{ass_PL}. In the case of categorical random variables,  the log-likelihood function \eqref{log-likelihood} for the sample $\bx^n=\{\bx^{(i)}\}_{i=1}^n$, evaluated at the maximum likelihood estimator for $\theta$ can be written as
\begin{equation}
l(C;\bx^n) = \sum_{j=1}^{k_C} \sum_{a\in A} N_{I_j}(a) \log \frac{N_{I_j}(a)}{n|I_j|}\,,
\end{equation}
where
\begin{equation}
N_{I_j}(a) \;=\; \sum_{i=1}^n\sum_{c=c_{j-1}+1}^{c_j}\mathbf{1}\{x^{(i)}_c=a\}\,.
\end{equation}
To check (PL1) in Assumption~\ref{ass_PL} let any set $C\in\C$. 
By the Law of Large Numbers we have that 
\[
\frac1n l(C;\bx^n) \;\to\; \sum_{j=1}^{k_C}  |I_j| \sum_{a\in A}  \hat p_{I_j}(a) \log \hat p_{I_j}(a) =: l^*(C)\,,
\]
with 
\[
\hat p_{I_j}(a) \;=\; \sum_{r=1}^{k_{C^*}}  \frac{|I_j\cap I^*_r|}{|I_j|} p_r^*(a) \qquad a\in A\,. \quad
\]
Moreover, if $C\not\supseteq C^*$ we have that 
\begin{equation}
\begin{split}
l^*(C) - l^*(C^*) &\;= \;   \sum_{j=1}^{k_C} \sum_{r=1}^{k_{C^*}} |I_j\cap I^*_r|  \sum_{a\in A}  p_r^*(a) \log \hat p_{I_j}(a)  -  \sum_{r=1}^{k_{C^*}} |I^*_r|  \sum_{a\in A}  p_r^*(a) \log p_r^*(a)\\
&=  \; \sum_{j=1}^{k_C} \sum_{r=1}^{k_{C^*}} |I_j\cap I^*_r|  \sum_{a\in A}  p_r^*(a) \log \frac{\hat p_{I_j}(a)}{p_r^*(a)}\\
&< \;0
\end{split}
\end{equation}
unless $\hat p_{I_j} = p_r^*$ for all $r$ and $j$ with $|I_j\cap I^*_r|\neq\emptyset$. But this can only happen if $C\supseteq C^*$, which is a contradiction. 
To verify hypothesis (PL2) observe that for any $C\supseteq C^*$
\begin{equation*}
\begin{split}
l(C;\bx^n)  - l(C^*;\bx^n) \;&=\; \sum_{j=1}^{k_C} \sum_{a\in A} N_{I_j}(a) \log \frac{N_{I_j}(a)}{n|I_j|} - \sum_{j=1}^{k_{C^*}} \sum_{a\in A} N_{I_j^*}(a) \log \frac{N_{I_j^*}(a)}{n|I_j^*|}\\
&\leq\; \sum_{j=1}^{k_C} \sum_{a\in A} N_{I_j}(a) \log \frac{N_{I_j}(a)}{n|I_j|} - \sum_{j=1}^{k_{C^*}}  \sum_{a\in A}  N_{I_j^*}(a)  \log p_j^*(a)\,,
\end{split}
\end{equation*}
where the last inequality follows because $N_{I_j^*}(a)/n|I_j^*|$ are the maximum likelihood estimators for $p_j^*(a)$.  As $C\supseteq C^*$, we have that any interval $I_r^*$ contains one or more intervals $I_j$, and 
\[
N_{I_r^*}(a) \; = \; \sum_{j: I_j\subseteq I_r^*}N_{I_j}(a) \qquad\text{ for all }a\in A\,.
\]
Then, we can write the difference above as
\begin{equation*}
    \begin{split}
\sum_{j=1}^{k_C} \sum_{a\in A} N_{I_j}(a) \log \frac{N_{I_j}(a)/n|I_j|}{p_{i_j}^*(a)}
    \end{split}
\end{equation*}
where $i_j$ is the corresponding index in $C^*$. Therefore we obtain that 
\begin{equation*}
\begin{split}
l(C;\bx^n)  - l(C^*;\bx^n) \;&\leq\; \sum_{j=1}^{k_C}  n |I_j|\, D(\widehat p_j ; p_{i_j}^*)\,,
\end{split}
\end{equation*}
where
\[
\widehat p_j(a) \;=\; \frac{N_{I_j}(a)}{n|I_j|}\,,\quad a\in A 
\]
and $D$ denotes the \emph{Kullback-Leibler} divergence between the probability distributions 
$\{\widehat p_j(a)\}_{a\in A}$ and $\{p^*_{i_j}(a)\}_{a\in A}$. 
By a well-known inequality, see for example  \citet[Lemma 6.3]{Csiszar2006} we have that 
\[
D(\widehat p_j(a) | p_{i_j}^*(a)) \;\leq\; \sum_{a\in A} \frac{|\widehat p_j(a) - p_{i_j}^*(a)|^2}{p_{i_j}^*(a)}\,.
\]
As the difference $N_{I_j}(a) - n|I_j|p_j(a)$ can be written as a sum of zero-mean independent random variables with finite variance, we have, by the Law of the Iterated Logarithm, see for example   \citet[Theorem 3.52]{breiman1969probability}, that 
\[
|\widehat p_j(a) - p_{i_j}^*(a)| \;\leq\; \sqrt{\frac{c\log\log n|I_j|}{n|I_j|}} \;\leq\;  \sqrt{\frac{c\log\log nm}{n|I_j|}}
\]
for a given constant $c$ and for all $a\in A$, eventually almost surely as $n\to\infty$. 
Hence for all $\delta > 0$ we have that 
\[
|\widehat p_j(a) - p_{i_j}(a)| \;\leq\; \sqrt{\frac{\delta \log n}{n|I_j|}}
\]
eventually almost surely as $n\to\infty$.
Finally, we obtain that 
\begin{equation*}
\begin{split}
l(C;\bx^n)  - l(C^*;\bx^n) \;&\leq\;  \frac{k_C\delta}{p_{\min}^*} \log n\;\leq\; \frac{m\delta}{p_{\min}^*}\log n \,,
\end{split}
\end{equation*}
with $p_{\min}^* = \min\{p_j^*(a) \colon p_j^*(a)>0,  j=1,\dots,k_{C^*}, a\in A\}$.  On the other hand, it is easy to see that 
\[
l(C;\bx^n)  - l(C^*;\bx^n) \;\geq\; 0
\]
for any $C\supseteq C^*$. Then, hypothesis (PL2) holds for $v(n)=\delta\log n$, for any $\delta>0$. This establishes the conditions on Assumption~\ref{ass_PL}.

Now we verify the hypotheses (H1)-(H2) in Assumption~\ref{ass_HS}. Observe that as in the case of (PL1), the Law of Large Numbers can be invoked to prove that
\begin{equation}
\frac1n l(I;\bx^n) \;\to\; l^*(I) \colon=  |I|\, \sum_{a\in A}  \hat p_{I}(a) \log \hat p_{I}(a)\,
\end{equation}
for any integer interval $I\in \mathcal I$. 
Now, for $I=r:s$ consider the function $h_I^*:I\to \mathbb{R}$ defined as
\begin{equation*}
h^*(u) \;=\; -l^*(r:u) - l^*((u+1):s)\,,\quad  u\,\in\,I\,.
\end{equation*}
As $l^*(I)$ satisfies the hypotheses of Lemma~\ref{conditions_iid_HS} with $\psi(\hat p) = \sum_{a\in A} \hat p(a) \log \hat p(a)$ we have that 
\[
\min_{u\in I\setminus\{s\}\cap  C^*} h^*(u)\;<\; \min_{u\not\in I\setminus\{s\}\cap C^*} h^*(u)
\]
whenever $I\setminus\{s\}\cap  C^*\neq\emptyset$, 
concluding the proof of (H1). The proof of (H2) follows by (PL2) in Assumption~\ref{ass_PL}.\hfill\qedhere
\end{proof}

\bibliographystyle{imsart-nameyear}
\bibliography{references}

\end{document}


\begin{frontmatter}
\title{Supplementary material to the article "POPULATION BASED CHANGE-POINT DETECTION FOR THE IDENTIFICATION OF HOMOZYGOSITY ISLANDS"}
\runtitle{Block change-point detection of population homozygosity islands}

\begin{aug}
\author[A]{\fnms{Lucas} \snm{Prates}},
\author[B]{\fnms{Renan} \snm{B. Lemes}}
\author[B]{\fnms{Tábita} \snm{Hünemeier}}
\and
\author[A]{\fnms{Florencia} \snm{Leonardi}}
\address[A]{Instituto de Matemática e Estatística, Universidade de São Paulo}

\address[B]{Instituto de Biociências, Universidade de São Paulo}
\end{aug}

\end{frontmatter}

\section*{Proof of Proposition 3}

For the Gaussian distribution, the maximum likelihood  estimator of the mean on an interval $I$ is given by 
\begin{align*}
\widehat{\mu}_I\;& =\; \frac{1}{n|I|} \sum_{c\in I}\sum_{i=1}^n x_c^{(i)} \,.
\end{align*}

Then the log-likelihood function defined in the main document by (2.4) for a change point set $C=(c_0,\dots, c_k)\in\C$ is given by 

\begin{equation*}
l(C; \bx^n) = -\frac{nm}{2}\log(2\pi\sigma_*^2) - \frac1{2\sigma_*^2}\sum_{j=1}^{k_{C}} \sum_{c\in I_j} \sum_{i=1}^n (x^{(i)}_c - \widehat\mu_{I_j})^2 \,,
\end{equation*}
where $I_j=(c_{j-1}+1):c_j$ for $j=1,\dots,k$. 
Notice that by the Law of Large Numbers, for all $j=1,\dots, k$ we have that 
\begin{align*}
\widehat{\mu}_{I_j} 
\;&=\; \sum_{r=1}^{k_{C^*}}\frac{|I_j \cap I_r^*|}{|I_j|}\,\left[\frac{1}{n|I_j \cap I_r^*|}\sum_{c\in I_j \cap I_r^*}\sum_{i=1}^n x_c^{(i)}\right]\quad
\overset{a.s.}{\longrightarrow}\quad  \sum_{r=1}^{k_{C^*}}\frac{|I_j \cap I_r^*|}{|I_j|} \mu_r^* \;=\colon\;  \bar\mu_{I_j} \,.
\end{align*}
Observe also that 
\begin{align*}
\sum_{j=1}^{k_{C}} \sum_{c\in I_j} \sum_{i=1}^n (x^{(i)}_c - \widehat\mu_{I_j})^2 \;&=\; 
\sum_{j=1}^{k_{C}} \sum_{r=1}^{k_{C^*}} \sum_{c\in I_j\cap I^*_r} \sum_{i=1}^n (x^{(i)}_c - \mu_r^*+\mu^*_r- \widehat\mu_{I_j})^2 \\
&=\; 
\sum_{j=1}^{k_{C}} \sum_{r=1}^{k_{C^*}} \sum_{c\in I_j\cap I^*_r} \sum_{i=1}^n (x^{(i)}_c - \mu_r^*)^2 + 2(x^{(i)}_c - \mu_r^*)(\mu^*_r- \widehat\mu_{I_j})\\
&\qquad\qquad\qquad\qquad\qquad+ (\mu^*_r- \widehat\mu_{I_j})^2\,.
\end{align*}
Then, as $n\to\infty$ we have that 
\[
\frac1n \sum_{j=1}^{k_{C}} \sum_{r=1}^{k_{C^*}} \sum_{c\in I_j\cap I^*_r} \sum_{i=1}^n (x^{(i)}_c - \mu_r^*)^2 \quad\overset{a.s.}{\longrightarrow}\quad \sum_{j=1}^{k_{C}} \sum_{r=1}^{k_{C^*}} |I_j\cap I^*_r|\sigma_*^2 \;=\; m\sigma_*^2\,.   
\]
In the same way
\[
\frac1n \sum_{j=1}^{k_{C}} \sum_{r=1}^{k_{C^*}} \sum_{c\in I_j\cap I^*_r} \sum_{i=1}^n 2(x^{(i)}_c - \mu_r^*)(\mu^*_r- \widehat\mu_{I_j}) \quad\overset{a.s.}{\longrightarrow}\quad 0
\]
and
\[
\frac1n \sum_{j=1}^{k_{C^*}} \sum_{r=1}^{k_{C^*}} \sum_{c\in I_j\cap I^*_r} \sum_{i=1}^n  (\mu^*_r- \widehat\mu_{I_j})^2  \quad\overset{a.s.}{\longrightarrow}\quad  \sum_{j=1}^{k_{C}} \sum_{r=1}^{k_{C^*}} |I_j\cap I_r^*| (\mu^*_r- \bar\mu_{I_j})^2\,.
\]
Then, for any $C\in\C$, by some manipulation we obtain that the almost sure limit of the maximum log-likelihood function is 
\begin{equation*}
\frac{1}{n}l(C; \bx) \; \overset{a.s.}{\longrightarrow}\;  -\frac{m}{2}\log(2\pi\sigma_*^2)-\frac{1}{2\sigma_*^2} \left[m\sigma_*^2 + \sum_{j=1}^{k_{C}}\sum_{r=1}^{k_{C^*}} |I_j\cap I^*_r|(\mu_r^*- \bar\mu_{I_j})^2\right] \;=\colon\;  l^*(C) \,.
\end{equation*}
We begin by verifying conditions (PL1) and (PL2) on Assumption~1.
Observe that  for any interval  $I\in\mathcal I$ we have that 
\begin{equation}\label{I}
\begin{split}
\sum_{r=1}^{k_{C^*}} |I\cap I^*_r|\,(\mu_r^*- \bar\mu_{I})^2 \;&=\; 
\sum_{r=1}^{k_{C^*}} |I\cap I^*_r|\,\mu_r^{*,2} - 2 \bar\mu_{I} \sum_{r=1}^{k_{C^*}} |I\cap I^*_r|\,\mu_r^* + |I|\,\bar\mu_I^2\\
&=\; \sum_{r=1}^{k_{C^*}} |I\cap I^*_r|\mu_r^{*,2} - |I|\, \bar\mu_I^2\,.
\end{split}
\end{equation}
Then the function $l^*(C)$ is given by 
\[
l^*(C) \;=\; -\frac{m}{2}[ \log(2\pi\sigma_*^2)+1] - \sum_{r=1}^{k_{C^*}} |I_r^*|\mu_r^{*,2} + \sum_{j=1}^{k_{C}} |I_j|\bar\mu_{I_j}^2\,.
\]
By the strict convexity of $x\mapsto x^2$ we have that 
\[
\sum_{j=1}^{k_{C}} |I_j|\bar\mu_{I_j}^2 \;\leq\; \sum_{j=1}^{k_{C}} |I_j| \sum_{r=1}^{k_{C^*}} \frac{|I_j\cap I_r^*|}{|I_j|}\mu_{r}^{*,2}\;=\;  \sum_{r=1}^{k_{C^*}} |I_r^*| \mu_{r}^{*,2}
\]
then 
\[
l^*(C) \;\leq\; l^*(C^*)\,.
\]
Moreover, the inequality is strict unless $C\supseteq C^*$, and then (PL1) follows. 

To prove (PL2) in Assumption 1, if we define the empirical variance of an interval $I$ as
\[
{\hat{\sigma}_I}^2 = \frac{1}{n|I_j|} \sum_{c \,\in\, I_j}\sum_{i=1}^n (X_c^{(i)} - \hat{\mu}_j)^2
\]
then the log-likelihood for any block over a finite sample can be written as
\[
l(C; \bx^n) = -\frac{nm}{2}\log(\sigma_*^2) - \frac{1}{2\sigma_*^2}\sum_{r = 1}^{k_C}\hat{\sigma}_{I_r}^2
\]

Given a change point set $C$, define $\mathcal{I}_r(C) = \{ j\,\in\, (1:{k_{C}}) | I_j \subseteq I_r^* \}$, the set of all indices whose blocks defined by $C$ are contained in $I_r^*$. If $C$ does not segment $I_r^*$, then there is only one index in $\mathcal{I}_r$ and the set that it indexes is exactly $I_r^*$. The log-likelihood difference to the true change point set is then
\[
l(C;\bx^n) - l(C^*;\bx^n) = \frac{n}{2\sigma_*^2}\sum_{r=1}^{k_{C^*}}\left[ |{I_r}^*|\hat{\sigma}_{I_r^*}^2 - \sum_{j\,\in\, \mathcal{I}_r(C)}|I_j|\hat{\sigma}_{I_j}^2\right] \quad.
\]

The empirical variance estimator for the whole block satisfies
\begin{align*}
    {\hat{\sigma}_{I_r^*}}^2 
    &= \frac{1}{n|I_r|} \sum_{c \,\in\, I_r^*}\sum_{i=1}^n (X_c^{(i)} - \hat{\mu}_{I_r^*})^2 \\
    &= \frac{1}{n|I_r|} \sum_{c \,\in\, I_r^*}\sum_{i=1}^n \left[\left(X_c^{(i)} - \mu_r^*\right)^2\right] - (\hat{\mu}_{I_r^*} - \mu_r^*)^2 \\
    &\leq \frac{1}{n|I_r|} \sum_{c \,\in\, I_r^*}\sum_{i=1}^n \left[\left(X_c^{(i)} - \mu_r^*\right)^2\right] \,.
\end{align*}
    
A similar calculation for the variance of the nested blocks yields
 \begin{align*}
    {\hat{\sigma}_{I_j}}^2 
    &= \frac{1}{n|I_j|} \sum_{c \,\in\, I_j}\sum_{i=1}^n (X_c^{(i)} - \hat{\mu}_{I_j})^2 \\
    &= \frac{1}{n|I_j|} \sum_{c \,\in\, I_j}\sum_{i=1}^n \left[\left(X_c^{(i)} - \mu_r^*\right)^2\right] - (\hat{\mu}_{I_j} - \mu_r^*)^2\,.
\end{align*}  

Define $j' = \underset{\mathcal{I}_r(C)}{\argmax} (\hat{\mu}_j - \mu_r^*)^2$, we have
\begin{align*}
    \frac{n}{2\sigma_*^2}\left(|I_r^*|\hat{\sigma}_r^2 - \sum_{j\,\in\, \mathcal{I}_r(C) } |I_j|\hat{\sigma}_j^2 \right)
    &= \frac{n}{2\sigma_*^2}\left( \sum_{j\,\in\, \mathcal{I}_r(C)} (\hat{\mu}_j - \mu_r^*)^2 \right)\\
    &\leq \frac{|I_r^*|n}{2\sigma_*^2}(\hat{\mu}_{j'} - \mu_r^*)^2\\
    &= \frac{|I_r^*|}{2|I_{j'}|^2\sigma_*^2}\left[\frac{1}{\sqrt{n}} \sum_{c \,\in\, I_{j'}}\sum_{i=1}^n \left(X_c^{(i)} - \mu_r^*\right) \right]^2\,.
\end{align*}

Define
\begin{equation*}
    S_{j^*} = \left[\frac{1}{\sqrt{n\log\log n}} \sum_{c \,\in\, I_{j'}}\sum_{i=1}^n \left(X_c^{(i)} - \mu_r^*\right)\right] \,.
\end{equation*}

By the law of the iterated logarithm, there exists $\epsilon_{j^*} > 1$ such that
\begin{equation*}
    \underset{n \rightarrow \infty}{\limsup}\, S_{j^*} < \epsilon_{j^*} \,.
\end{equation*}
    
By $\limsup$ properties,
\begin{equation*}
    \underset{n \rightarrow \infty}{\limsup}\, S_{j^*}^2 < \epsilon_{j^*}^2 \,.
\end{equation*}

Dividing the terms inside each block by $\log\log n$, we have
\begin{align*}
    \underset{n\rightarrow\infty}{\limsup}\, \frac{|I_r^*|}{2|I_{j^*}|^2\sigma_*^2}S_{j^*}^2
    &\leq \frac{|I_r^*|}{2|I_{j^*}|^2\sigma_*^2}  \underset{n\rightarrow\infty}{\limsup}\,\left(S_{j^*}^2\right)\\
    &\leq \frac{|I_r^*|}{2|I_{j^*}|^2\sigma_*^2}\epsilon_j^2\,.
\end{align*}

Finally
\begin{align*}
    \frac{1}{\log\log n}l(C; \mathbf{X}) - l(C^*;\mathbf{X})
    &\leq \underset{n\rightarrow\infty}{\limsup}\, \frac{1}{\log\log n}l(C; \mathbf{X}) - l(C^*;\mathbf{X}) \\
    &\leq \sum_{r=1}^{k_{C^*}}\frac{|I_r^*|}{2|I_{j^*}|^2\sigma_*^2}\epsilon_j^2\\
    &\leq \epsilon \,,
\end{align*}
for some $\epsilon > 0$. This establishes condition (PL2).

We now verify conditions (H1) and (H2) on Assumption~2. Observe that 
for any interval  $I=r:s\in\mathcal I$ we have that 
\begin{align*}
\sum_{r=1}^{k_{C^*}} |I\cap I^*_r|\,(\mu_r^*- \bar\mu_{I})^2 \;&=\; 
\sum_{r=1}^{k_{C^*}} |I\cap I^*_r|\,\mu_r^{*,2} - 2 \bar\mu_{I} \sum_{r=1}^{k_{C^*}} |I\cap I^*_r|\,\mu_r^* + |I|\,\bar\mu_I^2\\
&=\; \sum_{r=1}^{k_{C^*}} |I\cap I^*_r|\mu_r^{*,2} - |I|\, \bar\mu_I^2\,.
\end{align*}
Then the function $h_I^*(u)$ in (4.2) is given by 
\[
h_I^*(u) \;=\; \alpha(I,\theta^*) + (u-r+1) \psi(\bar\mu_{r:u}) + (s-u) \psi(\bar\mu_{(u+1):s})\,,
\]
with $\psi(\mu) = \mu^2$, $\mu\in \mathbb R$. As $h^*_I$ satisfies the conditions of Lemma~2 we have that (H1) follows. The proof of (H2) in Assumption~2 is obtained in the same way as (PL2) in Assumption~1.  

\section*{Simulations for Gaussian random variables}

We  provide a similar simulation study as the one presented in the main text the Gaussian family of distributions.

In these plots again we observe over-segmentation for the hierarchical algorithm for small sample sizes, and more common than in the Bernoulli case. The dynamic programming algorithm seems again to be more robust with respect to the choice of the penalization constant $\lambda$.

\begin{figure}
    \centering
    \includegraphics[width = 0.49\textwidth]{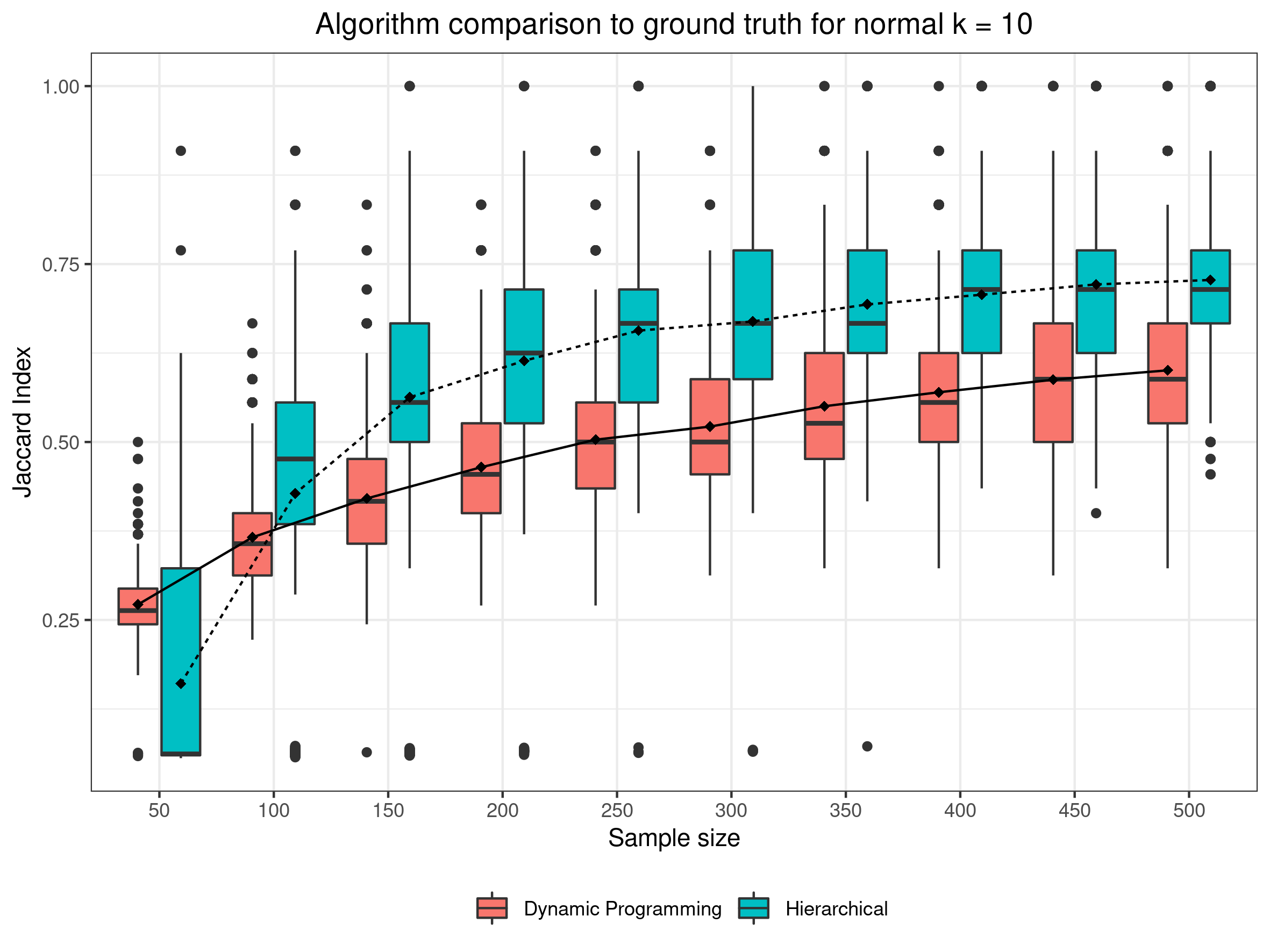}
    \includegraphics[width = 0.49\textwidth]{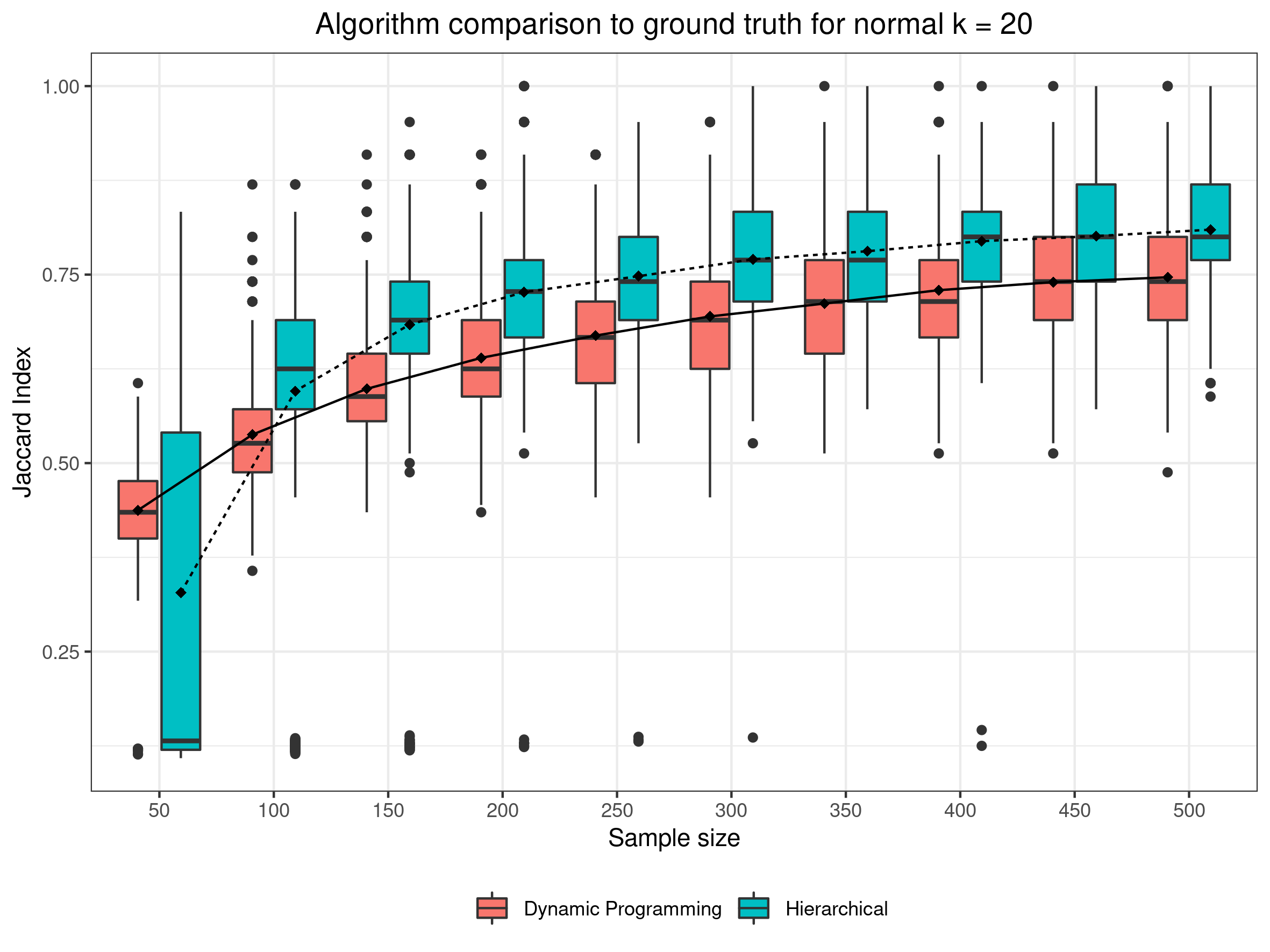}
    \includegraphics[width = 0.49\textwidth]{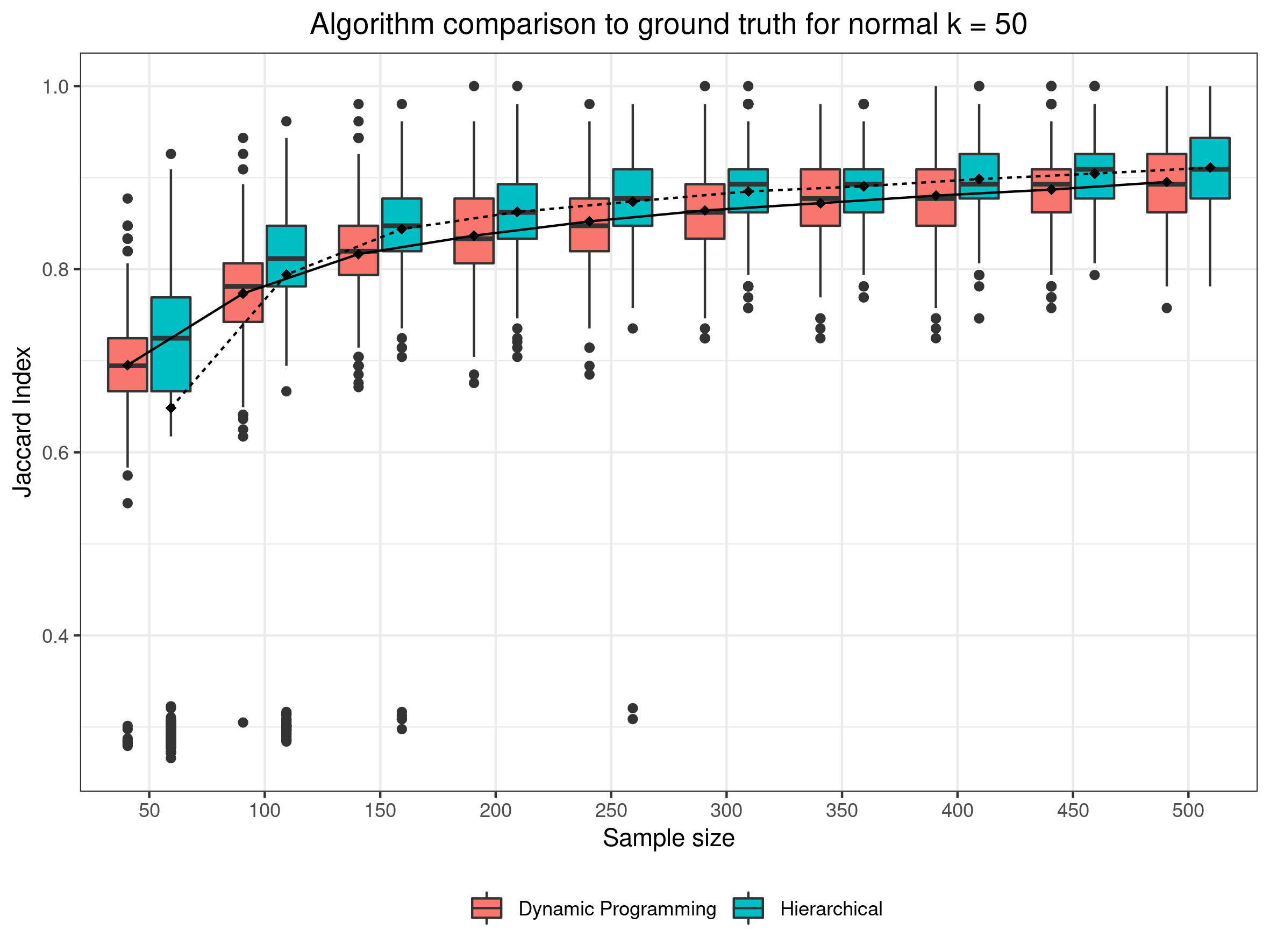}
    \caption{Comparison of estimated change point sets between algorithms and against ground truth for the Gaussian case using the Jaccard Index. The number of change points $k^*$ is indicated in the title.}
    \label{fig:metric_gaussian}
\end{figure}

In Figure \ref{fig:metric_gaussian} is evident that the hierarchical algorithm is more volatile, but as sample size grows, has an equivalent, and maybe even better, performance than the dynamic programming exact solution.

\begin{figure}[!htbp]
    \centering
    \includegraphics[width = 0.49\textwidth]{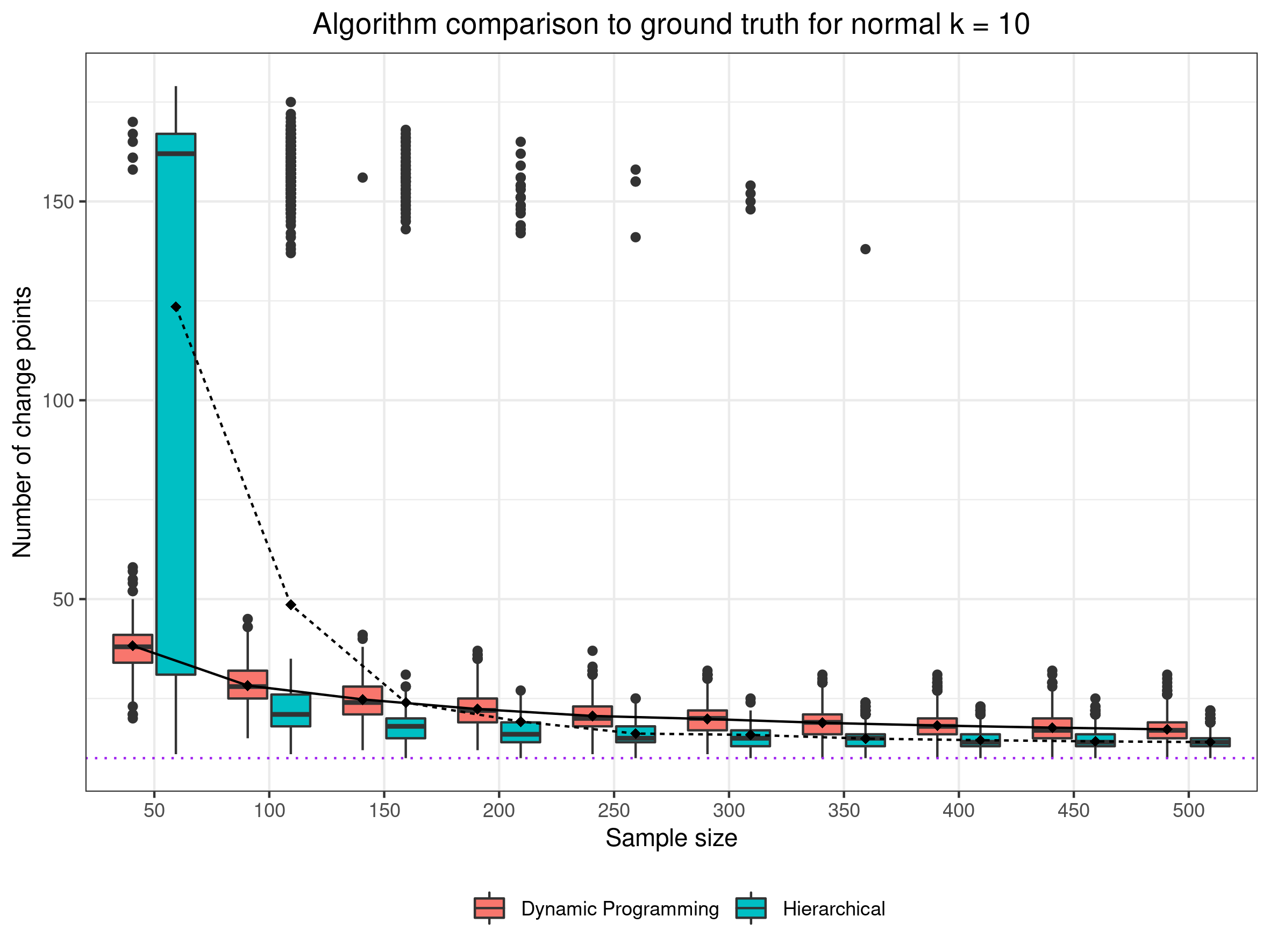}
    \includegraphics[width = 0.49\textwidth]{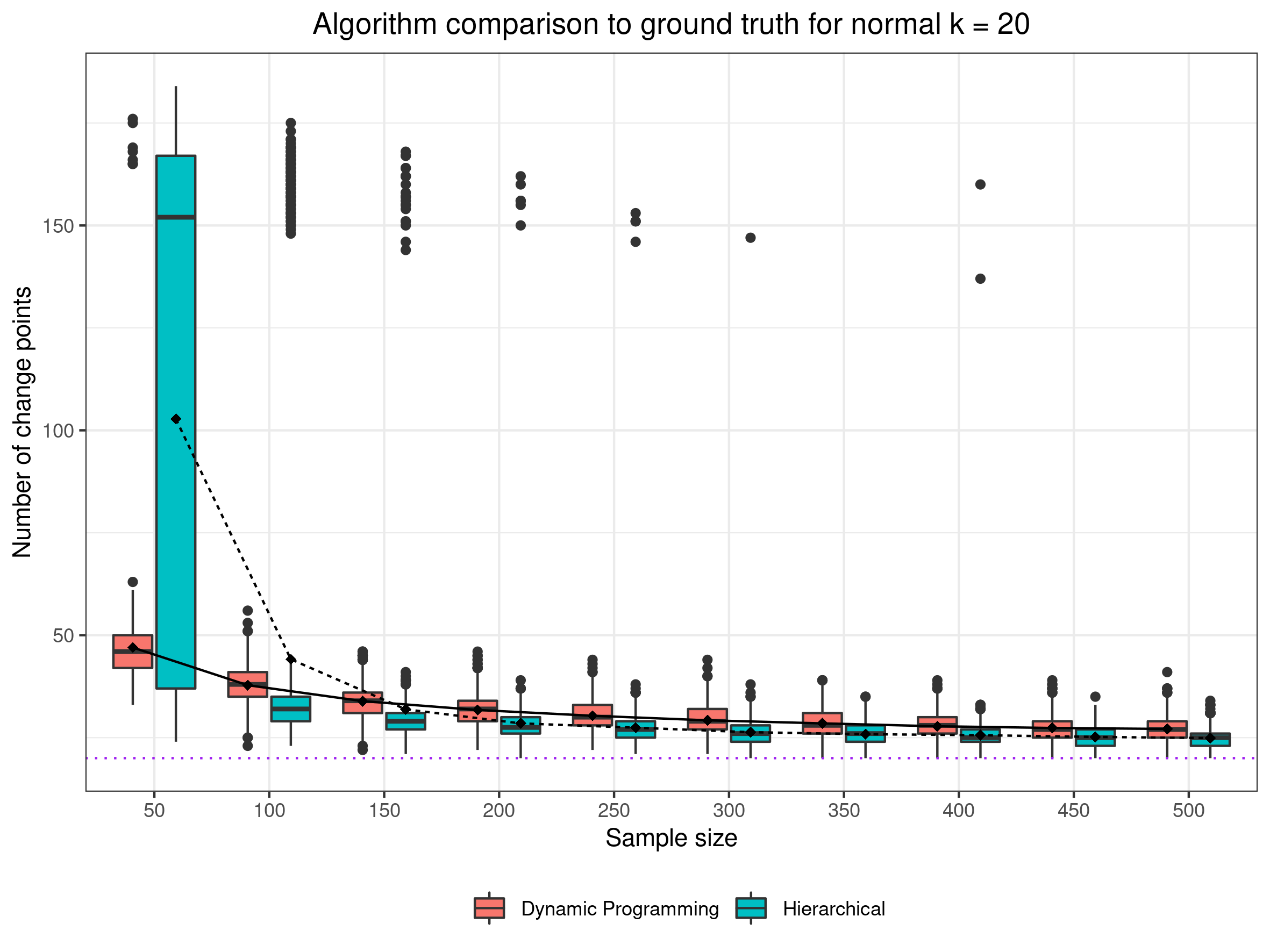}
    \includegraphics[width = 0.49\textwidth]{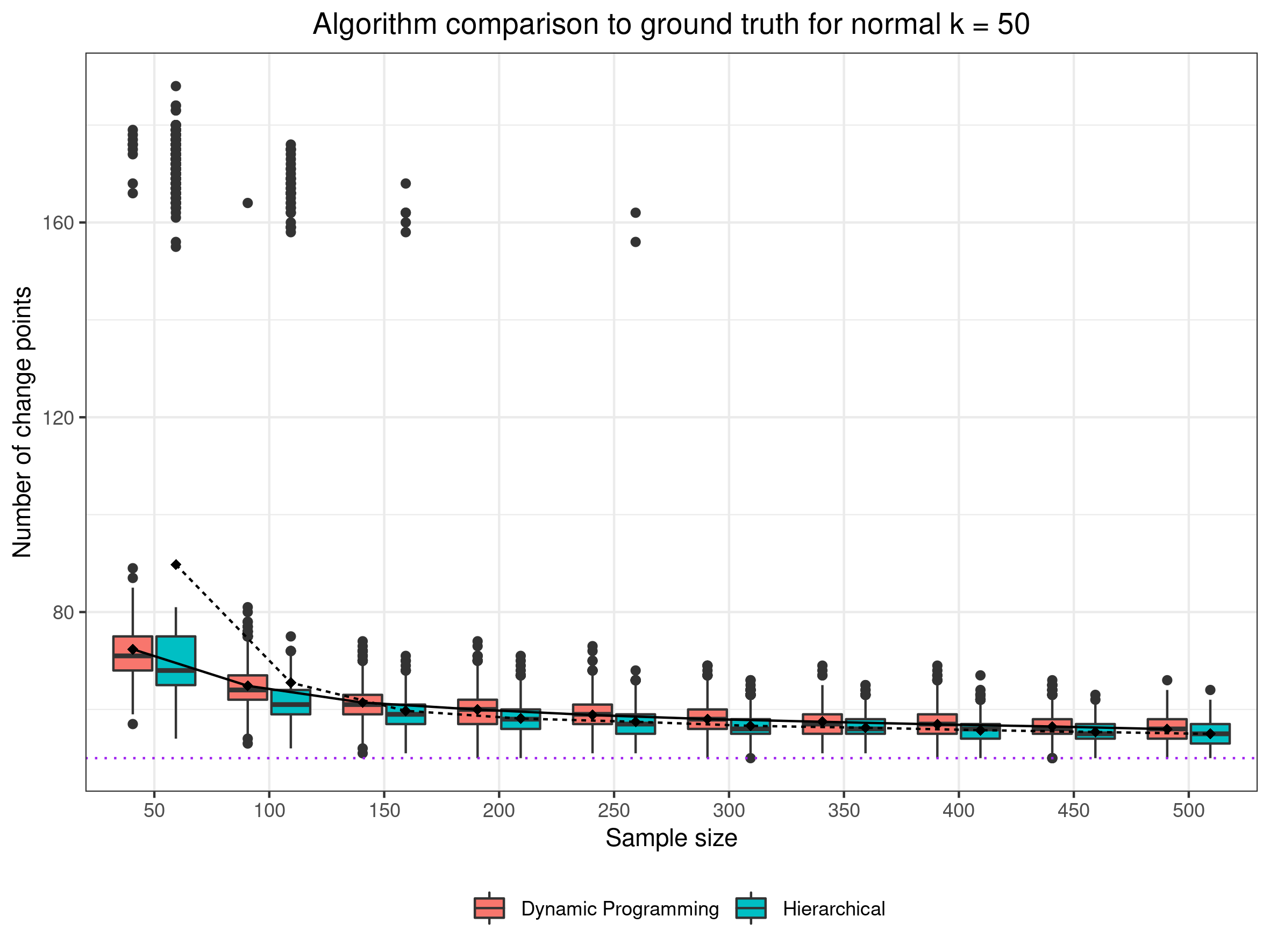}
    \caption{Comparison of estimated change point sets between algorithms and against ground truth for the Gaussian case using the estimated number of change points. The true number of change points $k^*$ is indicated in the title and marked on the dashed horizontal red line.}
    \label{fig:n_cp_gaussian}
\end{figure}